\documentclass[a4paper, USenglish,cleveref, autoref, thm-restate, numberwithinsect]{lipics-v2021}

\hideLIPIcs  

\usepackage{xspace}

\newcommand{\ie}{i.\,e.,\xspace}

\newcommand{\hcal}{\ensuremath{\mathcal{H}}\xspace}

\newcommand{\rcal}{\ensuremath{\mathcal{R}}\xspace}
\newcommand{\gcal}{\ensuremath{\mathcal{G}}\xspace}
\newcommand{\ecal}{\ensuremath{\mathcal{E}}\xspace}

\newcommand{\tuple}[1]{\ensuremath{\langle {#1} \rangle}\xspace}

\newcommand{\lifetime}{\ensuremath{\tau}\xspace}
\newcommand{\reachG}[1]{\ensuremath{\mathcal{R}(#1)}\xspace}

\newcommand{\setting}[1]{{\normalfont\texttt{#1}\xspace}}
\newcommand{\noTransform}[3]{\setting{#1} $\not\rightsquigarrow^{#2}$ \setting{#3}\xspace}
\newcommand{\noTransformAddOn}[5]{\setting{#1} {#2} $\not\rightsquigarrow^{#3}$ \setting{#4} {#5}\xspace}
\newcommand{\yesTransform}[3]{\setting{#1} $\rightsquigarrow^{#2}$ \setting{#3}\xspace}

\usepackage{lipsum}

\usepackage{tabularray}
\usepackage{colortbl}
\usepackage{multirow}
\usepackage{amsmath}
\usepackage[dvipsnames]{xcolor}
\usepackage{amssymb}
\usepackage{pifont}

\usepackage{hyperref}
\usepackage{listings}
\usepackage{float}

\usepackage{todonotes}
\newcounter{sharedquestioncounter} 
\newtheorem{openquestion}[sharedquestioncounter]{Open Question}
\newtheorem{question}[openquestion]{Question}

\usepackage[numbers]{natbib}
\title{Simple, Strict, Proper, and Directed: Comparing Reachability in Directed and Undirected Temporal Graphs} 

\titlerunning{Simple, Strict, Proper, and Directed: Reachability in (Un)Directed Temporal Graphs} 

    \author{Michelle {D\"oring}}{
        Hasso Plattner Institute, University of Potsdam, Potsdam, Germany
        \and \url{https://hpi.de/friedrich/people/michelle-doering.html} }{michelle.doering@hpi.de}{https://orcid.org/0000-0001-7737-3903}{German Federal Ministry for Education and Research (BMBF) through the project ``KI Servicezentrum Berlin Brandenburg'' (01IS22092)}
    
    \authorrunning{M. Döring} 
    
    \Copyright{Michelle Döring} 


\category{} 

\relatedversion{} 

\acknowledgements{I am deeply grateful to Arnaud Casteigts for his advice and inspiring discussions.} 

\nolinenumbers 

\EventEditors{John Q. Open and Joan R. Access}
\EventNoEds{2}
\EventLongTitle{42nd Conference on Very Important Topics (CVIT 2016)}
\EventShortTitle{CVIT 2016}
\EventAcronym{CVIT}
\EventYear{2016}
\EventDate{December 24--27, 2016}
\EventLocation{Little Whinging, United Kingdom}
\EventLogo{}
\SeriesVolume{42}
\ArticleNo{23}

\usepackage{xcolor}

\newcommand{\D}{D\xspace}
\newcommand{\UD}{UD\xspace}

\newif\ifshort

\begin{document}

\maketitle
    
    
\begin{abstract}
    We present the first comprehensive analysis of temporal settings for directed temporal graphs, fully resolving their hierarchy with respect to support, reachability, and induced-reachability equivalence. These notions, introduced by Casteigts, Corsini, and Sarkar [TCS24], capture different levels of equivalence between temporal graph classes.
    Their analysis focused on undirected graphs under three dimensions: \textit{strict vs. non-strict} (whether times along paths strictly increase), \textit{proper vs. arbitrary} (whether adjacent edges can appear simultaneously), and \textit{simple vs. multi-labeled} (whether an edge can appear multiple times). In this work, we extend their framework by adding the fundamental distinction of \textit{directed vs. undirected} edges.
    
    Our results reveal a single-strand hierarchy for directed graphs, with strict \& simple being the most expressive class and proper \& simple the least expressive. In contrast, undirected graphs form a two-strand hierarchy, with strict \& multi-labeled being the most expressive and proper \& simple the least expressive. The two strands are formed by the non-strict \& simple and the strict \& simple classes, which we show to be incomparable, thereby resolving the open question from [TCS24].
    
    In addition to examining the internal hierarchies of directed and of undirected graph classes, we compare the two directly.
    We show that each undirected class can be transformed into its directed counterpart under reachability equivalence, while no directed class can be transformed into any undirected one. Additionally, strict temporal graph classes, whether directed or undirected, are strictly more expressive than all non-strict classes.

    Our findings have significant implications for the study of computational problems on temporal graphs.
    Positive results in more expressive graph classes extend to all weaker classes as long as the problem is independent under reachability equivalence. Conversely, hardness results for a less expressive class propagate upward to all stronger classes. 
    We hope these findings will inspire a unified approach for analyzing temporal graphs under the different settings.

\end{abstract}

\section{Introduction}

A graph is called \textit{temporal} if its connections change over time.
Across disciplines, this concept has been given a myriad of names, including highly dynamic, edge-scheduled, multistage, time-labeled, and time-varying graph or network.
Unlike classical \textit{static} graphs, where every edge is always available, temporal graphs restrict the availability of each edge to specific points in time, making them ideal for modeling dynamic interactions and processes.

Real-world temporal networks contain either directed or undirected connections. Undirected connections commonly arise in social networks and peer-to-peer systems and can also model biological systems such as protein interaction networks and power grids. 
In contrast, directed connections are crucial for modeling one-way routes in transportation networks, logistic flows, and directed communication in email systems or online platforms. They also appear in ecological systems, such as predator-prey relationships. 
Notably, directed networks not only extend the complexity of undirected ones (replace each undirected edge with two opposing directed edges) but also introduce unique challenges, such as directed cycles.

In the context of this paper, a temporal graph is defined as $\gcal=(V,E,\lambda)$, with a finite set of vertices $V$, a set of edges $E\subseteq V\times V$, and a function $\lambda\colon E\rightarrow 2^\mathbb{N}$ assigning a set of time labels to every edge, indicating when the edge is present.
    Real-world applications have led to specialized temporal graph models, extending this basic framework by features, such as waiting times, dynamic edge capacities, and periodic or probabilistic edge appearances (refer to \cite{casteigts_time-varying_2012} for an overview).
    In this work, we focus on the basic temporal graph model.

Following the direction of Casteigt, Corsini, and Sakar (CCS) \cite{casteigts_simple_2024}, we investigate how simple changes in the definition of temporal graphs impact the reachability. 
A vertex can \textit{reach} another, if there exists a temporal path\,--\,a journey traversing edges in temporal order.
Reachability is a fundamental concept that has been the focus of most of the recent research on temporal graphs.
In particular the distinction between static reachability, which is transitive (if $x$ reaches $y$, and $y$ reaches $z$, then $x$ reaches $z$), and temporal reachability, which is not, has been widely studied and discussed \cite{casteigts_et_al2024distancetotransitivity, mertzios_complexity_2023}.
Reachability occurs in a variety of computational problems such as computing optimal journeys \cite{fuchsle_delay-robust_2022,kempe_connectivity_2002,xuan_computing_2002}, multi-agent pathfinding \cite{klobas_interference-free_2023,kunz_which_2023}, modeling infection or gossip spreading \cite{bumby1981problem,DELIGKAS2024105171,gobel_label-connected_1991}, finding connected components \cite{bhadra2003complexity,Costa2023OnCL}, computing flows or separators \cite{berman1996vulnerability,fluschnik_temporal_2020,vernet_theoretical_2021}, and designing temporal spanning graphs \cite{casteigts_search_2023,casteigts_sharp_2022}.

    Our work builds on the framework of temporal graph settings introduced by CCS, which was previously applied only to undirected graphs and identified three key dimensions within the \emph{settings} of a temporal graph:
    \begin{bracketenumerate}
        \item \textit{strict} vs \textit{non-strict} (whether times along a path are strictly increasing or non-decreasing);
        \item \textit{proper} vs \textit{arbitrary} (whether adjacent edges are forbidden to appear at the same time);
        \item \textit{simple} vs \textit{multi-labeled} (whether edges can appear at most once or multiple times).
    \end{bracketenumerate}
    We extend the framework by introducing a new dimension\,--\,one of the primary decisions that must be made when analyzing any graph:
    \begin{bracketenumerate}
        \setcounter{enumi}{-1}
        \item \textit{directed} vs \textit{undirected}.
    \end{bracketenumerate}
    This distinction is fundamental, as it directly shapes the definitions of graph properties and the formulation of computational problems, both in temporal and static graph theory.
    Throughout this paper, we refer to a setting and its corresponding class of temporal graphs by their defining properties in type-font; for instance, the directed, strict, and proper graphs are denoted by \setting{directed \& strict \& proper}.
    
    The different classes of temporal graphs are compared under four equivalence relations: bijective, support, reachability, and induced-reachability equivalence.
    These relations define different levels of comparability, based on the \textit{reachability graph} $\rcal(\gcal)$\,--\,a static graph with an edge from $u$ to $v$ if and only if there is a temporal path from $u$ to $v$ in \gcal.
    We particularly focus on reachability equivalence, which holds for two temporal graphs if and only if they share the same reachability graph (up to renaming the vertices).
\begin{figure}[t]
    \centering
    \makebox[\textwidth][c]{\includegraphics[width=1\textwidth]{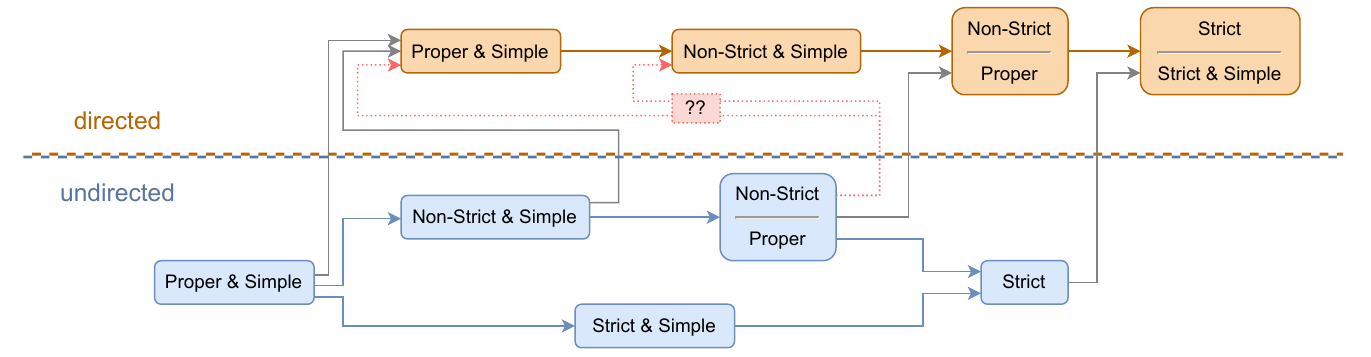}}
    \caption{Illustration of the reachability hierarchy for the temporal graph settings. An arrow from (1) to (2) indicates that every graph in (1) can be transformed into a graph in (2) with the same reachability graph. Indirect transformations are implied. 
    The two red dotted arrows 
    represent the open questions.
    The orange boxes (top) show the directed hierarchy (\Cref{sec:directed-reachability}) and the blue boxes (bottom) the undirected hierarchy (\Cref{sec:undirected-reachability}).
    The merged hierarchy is discussed in \Cref{sec:directed-VS-undirected}.}
    \label{fig:merged-hierarchy1}
\end{figure}

\paragraph*{Our Contribution} Our contribution is three-fold. First, we introduce and fully resolve the directed hierarchy. Second, we answer the open questions from \cite{casteigts_simple_2024}, thereby completing the undirected hierarchy. Third, we compare the directed and undirected graph classes to obtain a merged hierarchy.

\vspace{0.2em} We fully resolve the directed hierarchy with respect to the different equivalence notions, placing particular emphasis on reachability equivalence.
The top half (orange) of \Cref{fig:merged-hierarchy1} illustrates the directed reachability hierarchy.
We begin by identifying subset relations among the graph classes based on their definitions.
Next, we introduce four transformations establishing further relations:
    \begin{enumerate}
        \item A simple structural observation shows that \setting{directed \& strict} graphs and \setting{directed \& strict \& simple} graphs are equivalent under reachability equivalence.
        \item \textit{Support-dilation}, a generalization of the dilation process from \cite{casteigts_simple_2024}, transforms \setting{directed \& non-strict} into \setting{directed \& proper}, preserving both the reachability graph and the underlying paths (support equivalence).
        \item A novel transformation, \textit{reachability-dilation}, provides a more efficient way to transform \setting{directed \& non-strict} into \setting{directed \& proper} while preserving reachability.
        \item A generalization of the semaphore construction from \cite{casteigts_simple_2024} demonstrates that all directed settings can produce equivalent reachability graphs, provided we allow the introduction of auxiliary vertices during the transformation (induced-reachability equivalence).
    \end{enumerate}
For the remaining setting comparisons, we present separating structures\,--\,reachability graphs that can be expressed in one setting but not in another. These structures, primarily composed of directed cycles, clearly highlight the structural differences between the directed settings.

\vspace{0.2em} We complete the undirected hierarchy by answering the two open questions from \cite{casteigts_simple_2024}. Formally, we give an \setting{undirected \& non-strict} graph which cannot be transformed into an \setting{undirected \& strict} graph while preserving reachability. Thus, the undirected reachability hierarchy contains two incomparable strands, illustrated in the bottom half (blue) of \Cref{fig:merged-hierarchy1}.

\vspace{0.2em} We compare the directed and undirected graph classes, nearly completing the merged reachability hierarchy, which is illustrated by the arrows between top and bottom in \Cref{fig:merged-hierarchy1}.
First, we identify a straightforward transformation from any undirected setting to its directed counterpart (replace each undirected edge with two opposing directed edges), while no transformation exists in the reverse direction. 
Next, we establish a clear separation between strict and non-strict settings by presenting an \setting{undirected \& strict} (\setting{\& simple}) graph that cannot be transformed into any directed setting weaker than \setting{directed \& strict} while preserving reachabilities.
Finally, we show that reachability-dilation transforms any \setting{undirected \& non-strict \& simple} graph into the weakest directed setting, \setting{directed \& proper \& simple}, which is less expressive than its directed analogue.
This leaves two open questions: Can \setting{undirected \& non-strict} graphs be transformed into \setting{directed \& proper \& simple} or \setting{directed \& non-strict \& simple}, or neither?
\section{Preliminaries} \label{sec:prelims}

\subsection{Temporal Graphs}
    A \textit{temporal graph} $\gcal=(V,E,\lambda)$ consists of a static graph $G=(V,E)$, called the \textit{footprint}, along with a labeling function $\lambda$.
    A pair $(e, t)$, where $e \in E$ and $t \in \lambda(e)$, represents a \textit{temporal edge} with \textit{label} $t$. The set of all temporal edges is denoted by $\ecal$. The range of $\lambda$ is referred to as the \textit{lifetime} \lifetime, and the static graph $G_t = (V, E_t)$, where $E_t = \{e \in E \mid t \in \lambda(e)\}$, is called the \textit{snapshot} of $\gcal$ at time $t$.
    A \textit{temporal path} is a sequence of temporal edges $\tuple{(e_i,t_i)}$ where $\tuple{e_i}$ forms a path in the footprint, and the time labels $\tuple{t_i}$ are non-decreasing. 
    We denote a temporal path from $u$ to $v$ by $u \rightsquigarrow v$, omitting the term "temporal" for simplicity.
\subsection{Reachability and Equivalence Relations}
    The reachability relation between vertices of a temporal graph \gcal 
    can be captured by the \textit{reachability graph} $\rcal(\gcal)$ which is given by the static directed graph $(V,E_c)$ with $(u,v)\in E_c$ if and only if there exists a temporal path from $u$ to $v$. A graph is \textit{temporally connected} if all vertices can reach each other by at least one path, \ie \reachG{\gcal} is a complete directed graph.
    A \textit{temporally connected component} of \gcal is a subset of vertices which are temporally connected.


    We now define the equivalence relations between temporal graphs in relation to reachability. The definitions are adopted from \cite[Section 3]{casteigts_simple_2024}. 
    Let $\gcal_1=(V_1,E_1,\lambda_1)$ and $\gcal_2=(V_2,E_2,\lambda_2)$ be temporal graphs and denote by $\mathbb{P}(\gcal_1)$ and $\mathbb{P}(\gcal_2)$ the sets of temporal paths in $\gcal_1$ and $\gcal_2$, respectively. We say two paths $P$ and $P'$ share the same \emph{support} if they visit the same vertices in the same order, \ie their underlying static paths are the same.
    
    \begin{definition}[Bijective equivalence]
        Two temporal graphs $\gcal_1$ and $\gcal_2$ are \emph{bijective equivalent} if $V_1=V_2$ and there is a bijection $\sigma:\mathbb{P}(\gcal_1)\rightarrow\mathbb{P}(\gcal_2)$
        , such that each $P\in\mathbb{P}(\gcal_1)$ and $\sigma(P)\in\mathbb{P}(\gcal_2)$ share the same support.
    \end{definition}
    \begin{definition}[Support equivalence]
        Two temporal graphs $\gcal_1$ and $\gcal_2$ are \emph{support equivalent} if $V_1=V_2$ and 
        for every temporal path $P$ in either graph, there exists a temporal path $P'$ in the other graph, such that $P$ and $P'$ share the same support. 
    \end{definition}
    \begin{definition}[Reachability equivalence]
        Two temporal graphs $\gcal_1$ and $\gcal_2$ are \emph{reachability equivalent} if $V_1=V_2$ and $\rcal(\gcal_1)\simeq \rcal(\gcal_2)$, \ie the reachability graphs are isomorphic.
    \end{definition}
    By slight abuse of language, we will say that $\gcal_1$ and $\gcal_2$ have the same reachability graph whenever $\rcal(\gcal_1)\simeq \rcal(\gcal_2)$.
    \begin{definition}[induced-Reachability equivalence]
        A temporal graph $\gcal_2$ is \emph{induced-reachability equivalent} to $\gcal_1$ if $V_1\subseteq V_2$ and $\rcal(\gcal_1)\simeq \rcal(\gcal_2)[V_1]$. 
    \end{definition}
    The four equivalence notions are sorted by decreasing requirement strength: bijective implies support implies reachability, which in turn implies induced-reachability equivalence.

    In this paper, we analyze whether a class of temporal graphs $\mathbb{G}_1$ can be transformed into another class $\mathbb{G}_2$ while preserving one of the equivalence notions: bijective, support, reachability, or induced-reachability. Formally, such a \textit{transformation} exists if, for every $G_1 \in \mathbb{G}_1$, there is a $G_2 \in \mathbb{G}_2$ such that $G_1$ and $G_2$ are equivalent under the chosen notion.
    This is denoted by \yesTransform{$\mathbb{G}_1$}{X}{$\mathbb{G}_2$}, where $X$ specifies the equivalence type ($B$, $S$, $R$, or $iR$).
    Inversely, a \emph{separating structure} or \emph{separation} with respect to an equivalence notion is a graph $G\in\mathbb{G}_1$ for which there exists no graph in $\mathbb{G}_2$ that is equivalent under the same notion.
    
    By definition, a transformation under a given equivalence notion guarantees a transformation for all weaker notions.
    Similarly, a separating structure under a specific notion also separates the same classes under every stronger notion.
    
    \subsection{Settings: (Un)Directed, (Non-)Strict, Simple, and Proper} \label{subsec:settings}
    We identify four key dimensions for defining the settings of a temporal graph \gcal, following \cite{casteigts_simple_2024}:
    \begin{enumerate}
    \item \textit{Directed vs. Undirected}: The temporal graph \gcal is called directed or undirected based on whether its footprint is a directed or undirected graph.
   
    \item \textit{Strict vs. Non-strict}: A temporal path is said to be \textit{strict} if the sequence of time labels is strictly increasing. Otherwise, it is called \textit{non-strict}. When reachability is considered using (non-)strict paths, we say \gcal is (non-)strict. 

    \item \textit{Proper vs. Arbitrary}: A temporal graph is called \textit{proper} if no two edges incident to the same vertex share a time label. 
       
    \item \textit{Simple vs. Multi-labeled}: A temporal graph is called \textit{simple} (also referred to as single-labeled) if each edge is assigned exactly one time label. If an edge can have multiple time labels, \gcal is \textit{multi-labeled}, which is typically the default setting.   
    \end{enumerate}
    A temporal graph \gcal is called \textit{happy} \cite{casteigts_simple_2024} if it is both proper and simple.

    We denote the class of temporal graphs belonging to, for instance, the directed, strict and simple setting, as \setting{directed \& strict \& simple}. By slight abuse of notation, we may refer to \gcal itself as a \setting{directed \& strict \& simple} graph rather than explicitly stating that \gcal belongs to the class of graphs defined by this setting.

    To further enhance readability, we will abbreviate any directed setting as \setting{\D \&} \dots\ and any undirected setting as \setting{\UD \&} \dots\ in the remainder of this paper.


    Observe that the settings \setting{strict \& proper} and \setting{non-strict \& proper} are equivalent to \setting{proper}: In any proper temporal graph, each pair of edges that could be taken in succession must have distinct time labels. Therefore, every temporal path has strictly increasing time labels, making the distinction between strict and non-strict paths irrelevant.
    Thus, we will omit the ``(non-)strict'' when referring to a \setting{proper} setting.

\begin{figure}[h]
    \centering
        \begin{tabular}{@{\hspace{-0.5em}}m{0.46\linewidth}@{\hspace{2em}}m{0.41\linewidth}@{}}
        \includegraphics[width=\linewidth]{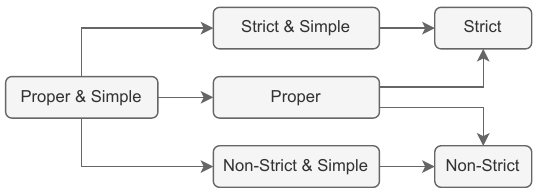}& 
        \includegraphics[width=\linewidth]{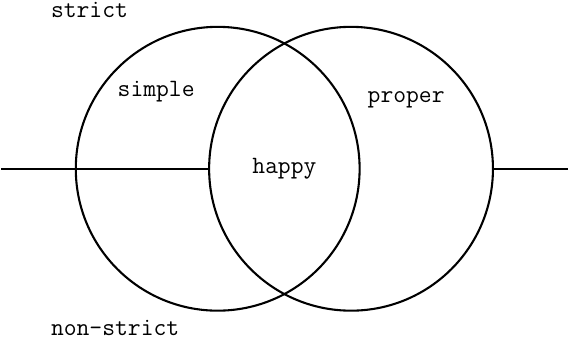}
        \end{tabular}
    \caption{On the left is an illustration of the subset relations between the settings. An arrow from (1) to (2) represents that (1) is a subset of (2). This holds for both directed and undirected graphs.
    On the right is an  illustration adopted from CCS, showing the subset relation in a venn diagram.}
    \label{fig:temporal-settingstwo}
\end{figure}
    In total there are 12 settings to consider, 6 for directed and 6 for undirected graphs. To analyze how these settings relate to one another with respect to the introduced equivalence notions, we first observe the subset relations implied by the setting definitions.
    If a setting is a subset of another\,—\,for example, \setting{\D \& proper} is a subset of \setting{\D \& non-strict}\,—\,there exists a trivial bijective-preserving transformation from \setting{\D \& proper} to \setting{\D \& non-strict} (the identity function).
    In \Cref{fig:temporal-settingstwo}, we visualize these subset relations by adapting \cite[Figure 1]{casteigts_simple_2024} on the right, and providing the hierarchy given by the subset relations on the left.
\section{Undirected Hierarchy -- Summary and Answer to Question 2 of \cite{casteigts_simple_2024}} \label{sec:undirected-reachability}
\begin{figure}[h]
    \centering
    \includegraphics[width=1\linewidth]{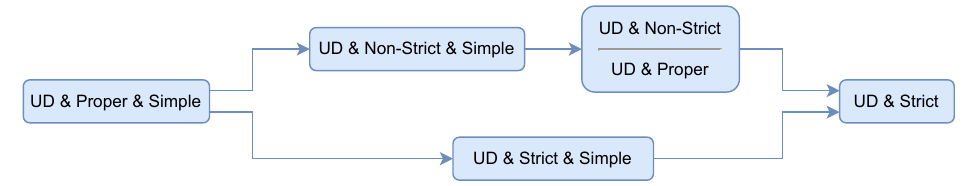}
    \caption{The reachability hierarchy of the undirected graph settings. The undirected hierarchy under support equivalence is the same. With respect to induced-reachability equivalence, all undirected settings are equivalent.}
    \label{fig:hierarchy-undirected}
\end{figure}
In this section, we give a brief overview of the results of CCS for undirected temporal graphs and address the last Open Question posed in \cite{casteigts_simple_2024}.
CCS identified three equivalence-preserving transformations for different equivalence notions, alongside the subset transformations. 

\emph{Dilation} transforms any \setting{\UD \& non-strict} graph into a support equivalent \setting{\UD \& proper} graph.
The process considers each snapshot individually. It assigns unique time labels to each edge within the snapshot by stretching the time step, duplicating each individual path occurring in that time step, and shifting all subsequent time labels accordingly.

\emph{Saturation} transforms any \setting{\UD \& non-strict} graph into a reachability equivalent \setting{\UD \& strict} graph. Similar to dilation, saturation considers each snapshot individually.
However, instead of preserving every individual path, it creates a direct edge with time label $t$ for every reachability achieved in time step $t$. Since \setting{\UD \& proper} is a subset of \setting{\UD \& strict}, the existence of a reachability equivalent transformation is already guaranteed by dilation. Nevertheless, saturation has the advantage of a significantly lower blow-up of the graph's lifetime, albeit at the cost of potentially producing a non-proper graph.

\emph{Semaphore} transforms any \setting{\UD \& strict} graph into an induced-reachability equivalent graph in the most restrictive \setting{\UD \& proper \& simple} setting.
To make every edge single-labeled, semaphore subdivides each temporal edge $(uv,t)$ by replacing it with $(ux_t,t)$ and $(x_tv,t)$, where $x_t$ is unique for each temporal edge.
To make the graph proper, semaphore processes each snapshot individually. Within a snapshot, every subdivided edge is duplicated, and the labels shifted so that one copy represents the direction $u$ to $v$, and the other represents $v$ to $u$. For further details on this transformation, refer to \Cref{def:semaphore}.
Since this process introduces new vertices, it can at most achieve induced-reachability equivalence. However, apart from the new vertices, the original paths are preserved, and in this sense, semaphore can be considered as induced-bijective-preserving.

Since \setting{\UD \& proper \& simple} is a subset of all other settings, and since dilation and semaphore allow all setting to be transformed into this setting under induced-reachability, we conclude that all undirected settings are equivalent with respect to induced-reachability.

For all other setting comparisons, apart from Open Question 1 and 2, CCS proved the existence of separating structures.
These structures are mostly composed of short paths and small triangles, and generally contain a small number of vertices.
%
Open Question 2 from \cite{casteigts_simple_2024} asked whether \setting{\UD \& non-strict \& simple} can be transformed into \setting{\UD \& strict \& simple}. In \Cref{thm:openQ}, we show that no such transformation exists. This also resolves Open Question 1 from the same work, which asked the equivalent question for \setting{\UD \& non-strict}.

Combining the results of \cite{casteigts_simple_2024} with \Cref{thm:openQ}, we obtain a two-strand hierarchy for undirected reachability, which is illustrated in \Cref{fig:hierarchy-undirected}.
\newcommand{\UDstatement}{UD\xspace}
\begin{theorem}[\noTransform{\UDstatement \& non-strict \& simple}{R}{\UDstatement \& strict \& simple}]\label{thm:openQ}
There is a graph in the \setting{\UDstatement \& non-strict \& simple} setting whose reachability graph cannot be obtained from a graph in the \setting{\UDstatement \& simple \& strict} setting.
\end{theorem}
\ifshort
\else
\begin{proof}
\setcounter{ppclaim}{0} 
    Consider the following temporal graph \gcal in the \setting{\UD \& non-strict \& simple} setting (left) and the corresponding reachability graph (right).
    \begin{figure}[h]
        \centering
        \includegraphics[width=0.9\linewidth]{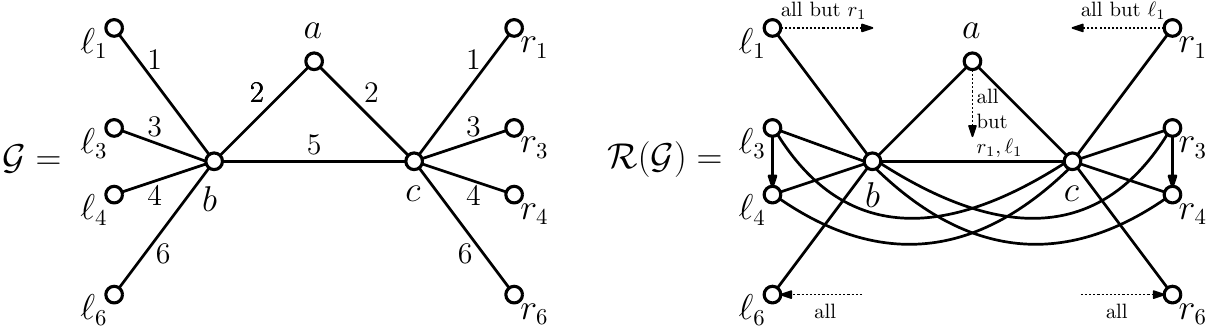}
        \label{fig:crab}
    \end{figure}
    For readabilities sake, the dotted edges adjacent to $\ell_1, r_1, a, \ell_6, r_6$ indicate incoming or outgoing edges that either connect to all vertices, or connect to all vertices except those explicitly excluded (specified next to the edge).
    For the sake of contradiction, let \hcal be a temporal graph in the \setting{\UD \& strict \& simple} setting whose reachability graph is isomorphic to that of \gcal.
    
    First, observe that there are only four undirected edges apart from the original edges of \gcal in the reachability graph.
    \begin{figure}[h]
        \centering
        \includegraphics[width=0.4\linewidth]{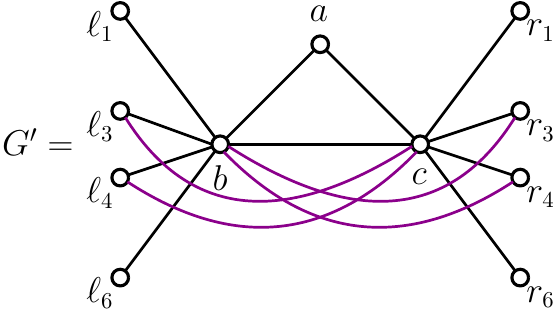}
        \label{fig:crab-additional-edges}
    \end{figure}
    That means, that the footprint of \hcal has to be isomorphic to a subgraph of the following graph $G'$, with the four edges highlighted in purple.
    We first prove an intermediate statement regarding the purple edges. 
    \begin{ppclaim}\label{claim:crab-purple-black-mutual-exclusion}
        The purple edge $c\ell_3$ can be in the footprint of \hcal if and only if no black edge from $\{cr_3,cr_4\}$ is in the footprint. The same holds for $c\ell_4$ and $\{cr_3,cr_4\}$, for $br_3$ and $\{b\ell_3,b\ell_4\}$, and for $br_4$ and $\{b\ell_3,b\ell_4\}$. 
    \end{ppclaim}
    \begin{claimproof}[Proof of \Cref{claim:crab-purple-black-mutual-exclusion}]
        We proof the claim for $c\ell_3$ and $\{cr_3,cr_4\}$; the other three cases follow by a symmetric argument.
        Assume $c\ell_3$ and $cr_3$ are in the footprint of \hcal. Then $\ell_3$ and $r_3$ have distance two in the footprint and by \cite[Lemma 1]{casteigts_simple_2024}, either of them must reach the other. But since $\ell_3r_3\notin E(\reachG{\gcal})$ and $r_3\ell_3\notin E(\reachG{\gcal})$, this yields a contradiction. The same argument holds for $cr_4$ in the footprint of \hcal, by \cite[Lemma 1]{casteigts_simple_2024} and $\ell_3r_4,r_4\ell_3\notin E(\reachG{\gcal})$.
    \end{claimproof}
    This leaves two cases for the footprint of \hcal: either we add no purple edge and $H$ is isomorphic to $G$, or we add at least one purple edge and delete the corresponding black edges.
    For the former, we show that there is no simple labeling yielding the desired reachabilities for \hcal.
    \begin{ppclaim} \label{claim:crab-caseone-no-purple}
        If the footprint of \hcal is isomorphic to the footprint of \gcal, then there can be no simple labeling achieving the same reachabilities in the strict setting.
    \end{ppclaim}
    \begin{claimproof}[Proof of \Cref{claim:crab-caseone-no-purple}]
        We show the claim by proving that the chronological order of the edges, \ie corresponding time labels, cannot change without changing the reachability in the strict setting. Since \hcal is simple, we also cannot add labels. The chronological order of the edges in $\mathcal{G}$ can be illustrated by the following partially ordered set. 
        \begin{figure}[ht]
            \centering
            \includegraphics[width=0.6\linewidth]{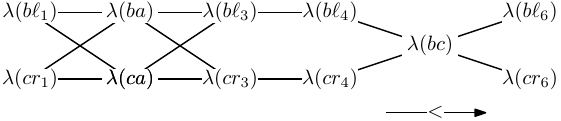}
        \vspace{-0.8em}\end{figure} 
    
        Now, since only $\ell_1$, $r_1$, $b$, $c$ can reach $a$, we infer that $\lambda(ba)$ and $\lambda(ca)$ must be strictly smaller than $\lambda(b\ell_3),\lambda(cr_3),\lambda(b\ell_4), \lambda(cr_4), \lambda(b\ell_6), \lambda(cr_6)$, but do not have to be the same.
        Furthermore, since $\ell_3$ can reach $\ell_4$ but not vice versa, we have $\lambda(b\ell_3)<\lambda(b\ell_4)$. Analogously, we have $\lambda(cr_3)<\lambda(cr_4)$.
        Now, for $r_3,r_4$ to reach $b$, and $\ell_3,\ell_4$ to reach $c$,  we need $\lambda(b\ell_4)<\lambda(bc)$ and $\lambda(cr_4)<\lambda(bc)$. 
        Lastly, since everyone can reach $\ell_6$ and $r_6$, we know that $\lambda(b\ell_6)$ and $\lambda(cr_6)$ must be greater than all other labels.
        Given all this, the only way for $\ell_1$ to reach everyone but $r_1$, and vice versa, is by making $\lambda(b\ell_1)$ and $\lambda(cr_1)$ strictly smaller than all other labels. Without the analysis about the ordering of $bc$ before, there would also be other ways for achieving the reachabilities of $\ell_1$ and $r_1$.

        This yields exactly the same chronological order for \hcal as it was the case for \gcal.    
        Now, since we cannot change the relative order of the time labels,  $\lambda(ba)$ or $\lambda(ca)$ has to change to enable the reachabilities of $\ell_1$ and $r_1$ in the strict setting. However, if $\lambda(ba)<\lambda(ca)$ then $r_1$ cannot reach $\ell_3$ and $\ell_4$, and if $\lambda(ba)>\lambda(ca)$ then $\ell_1$ cannot reach $r_3$ and $r_4$. Therefore, such a labeling is not possible in the \setting{\UD \& strict \& simple} setting.
    \end{claimproof}
    Therefore, $\reachG{\gcal}$ is not achievable for \hcal if its footprint is isomorphic to $G$.
    So, assume that at least one purple edge is added, without loss of generality, $c\ell_3$, and the corresponding black edges $cr_3$ and $cr_4$ are removed.
    To preserve the reachabilities of $r_3$ and $r_4$, they must remain connected to the graph, which can only be done by adding $br_3$ and $br_4$. This, however, implies that we need to remove the corresponding black edges $b\ell_3$ and $c\ell_4$.
    As a result, $\ell_4$ becomes disconnected, and we are forced to add $c\ell_4$.
    Thus, adding even one purple edge to the footprint requires adding all four purple edges.
    However, adding all purple edges and deleting the corresponding black edges results in a footprint that is isomorphic to $G$, with only $\ell_3$ and $\ell_4$ swapping places with $r_3$ and $r_4$, respectively. Thus, by \Cref{claim:crab-caseone-no-purple}, there is no simple labeling for \hcal yielding the desired reachability graph.
\end{proof}
\fi
\section{Directed Reachability} \label{sec:directed-reachability}
\begin{figure}[ht]
    \centering
    \includegraphics[width=\linewidth]{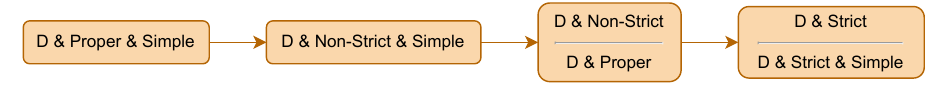}
    \caption{The reachability hierarchy of the directed graph settings. The directed hierarchy under support equivalence is different; it is the same hierarchy as in \Cref{fig:hierarchy-undirected}. With respect to induced-reachability equivalence, all directed settings are equivalent.}
    \label{fig:directed-hierarchy}
\end{figure}
In this section, we present our results on the equivalence hierarchy of directed temporal graphs.
While there are similarities to the results for undirected graphs\,—\,for example in the approaches for transformations\,—\,directed temporal graphs also have distinct differences in their separating structures and the resulting hierarchies.


For support equivalence, a generalization of \emph{Dilation} achieves \yesTransform{\D \& non-strict}{S}{\mbox{\D \& proper}}.
For all other setting comparisons under support equivalence, there exist separating structures involving one or more overlapping directed triangles. Therefore, under support equivalence, the directed settings have the same hierarchy as the undirected settings.

For reachability equivalence, there exists a transformation which introduces a notable contrast to the hierarchy of undirected graphs: \textbf{every} directed setting can be transformed into the \setting{strict \& simple} setting.
We also present a transformation similar to dilation which yields a reachability-preserving transformation from \setting{\D \& non-strict} to \setting{\D \& proper}. This process constructs a graph with at most twice as many labels, which is significantly smaller than the blow-up of dilation on undirected graphs.
On the negative side, the directed triangle still separates the \setting{strict} setting from the \setting{proper} and from the \setting{non-strict} settings.
Furthermore, the separation from \setting{proper} to \setting{non-strict \& simple} still holds.
As a result, the reachability hierarchy for directed graphs is as follows. 

Lastly, induced-reachability equivalence permits a directed analogue of the \emph{Semaphore} construction, which transforms the \setting{\D \& strict} setting to the \setting{\D \& proper \& simple} setting.
By combining dilation and semaphore, we can transform \textbf{any} temporal graph into an induced-reachability equivalent \setting{\D \& proper \& simple} graph.
Consequently, all directed settings are induced-reachability equivalent, as is the case for undirected graphs.

Before formally presenting our results, we establish the following two fundamental lemmas: 
the first concerns the non-strict reachability of a directed triangle, and the second addresses the labeling of a proper, fully connected graph.
\begin{lemma}\label{lem:nonstrict-triangles}
    Cycles in \setting{\D \& non-strict} temporal graphs always contain at least one transitive reachability. In particular, the reachability graph cannot contain an induced cycle. 
\end{lemma}
\begin{proof}
    Assume towards contradiction that there is a temporal graph $\gcal=(G,\lambda)$ whose reachability graph in the \setting{\D \& non-strict} setting contains an induced cycle on the vertices $v_1,\dots,v_n$, indexed in the order of the cycle. Then the induced footprint $G[v_1,\dots,v_n]$ has to be a cycle as well.
    We can conclude that $\lambda(v_1,v_2)>\lambda(v_2,v_3)$, as otherwise $v_1$ can reach $v_3$. Following this argument, we get $\lambda(v_1,v_2)>\dots>\lambda(v_{n-1},v_n)>\lambda(v_n,v_1)>\lambda(v_1,v_2)$, which is a contradiction. If we add more than one label to any edge, we only increase the reachabilities. Therefore, there can  be no such temporal graph \gcal.
\end{proof}
Note that in a \setting{\D \& strict} graph this can be easily achieved by taking the cycle as the footprint and assigning a uniform label, say $1$, to every edge, see \Cref{fig:directed-triangle}.
\begin{lemma}\label{lem:proper-cliques}
    Any \setting{\D \& proper} temporal graph on $n>2$ vertices whose reachability graph is a clique (all vertices are pairwise connected) has to contain at least $n+1$ temporal edges.
\end{lemma}
\begin{proof}
    Assume towards contradiction that there is a \setting{\D \& non-strict} temporal graph $\gcal=(V,E,\lambda)$ whose reachability graph is a clique but with $\lvert\ecal\rvert\leq n$.
    
    First, observe that there have to be at least $n$ static edges in \gcal, as otherwise the footprint of \gcal could not be strongly connected \cite{garcia2012minimal}, and \gcal could not be temporally fully connected. Thus, assume $\lvert\ecal\rvert=n$.
    For the footprint of \gcal to be strongly connected with $n$ edges, it must form a directed cycle.
    Let $v_1,\dots,v_n$ be the vertices indexed in the order of the cycle.

    Since \(\gcal\) is proper, every temporal path must use edges with strictly increasing time labels. To achieve temporal connectivity, vertex \(v_1\) must be able to reach all other vertices \(v_2, \dots, v_n\); in particular, there must be a temporal path from \(v_1\) to \(v_n\) along the edges of the cycle. 
    The same arguments holds for every pair of vertices on the cycle.
    However, with only \(n\) temporal edges and the properness condition, it is impossible to label the edges in a way that forms strictly increasing temporal paths between all vertex pairs. 
\end{proof}
We now present the reachability separations first, followed by the support separations, and lastly the transformations. Recall that any reachability separation also implies a support separation for the corresponding settings.

\subsection{Reachability separations}
The primary separating structure is the directed triangle with the same label on every edge.
For this graph in the \setting{\D \& strict \& simple} setting, there is no reachability equivalent graph in the \setting{\D \& non-strict} setting. This follows directly from \Cref{lem:nonstrict-triangles}.
\begin{lemma}[\noTransform{\D \& strict \& simple}{R}{\D \& non-strict}] \label{lem:strict-noR-nonstrict}
    There is a graph in the \mbox{\setting{\D \& strict \& simple}} setting whose reachability graph cannot be obtained from a graph in the \setting{\D \& non-strict} setting.
\end{lemma}
\begin{figure}[ht]
    \centering
    \includegraphics[width=0.5\linewidth]{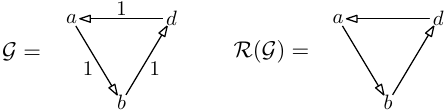}
    \caption{Directed triangle with time label 1 on each edge (left) and its reachability graph in the \setting{\D \& strict} (\setting{\& simple}) setting (right).}
    \label{fig:directed-triangle}
\end{figure}
Recall 
that \setting{\D \& strict \& simple} is a subset of \setting{\D \& strict}, 
\setting{\D \& non-strict \& simple} is a subset of \setting{\D \& non-strict}, and 
\setting{\D \& proper \& simple} is a subset of \setting{\D \& proper}, which is a subset of \setting{\D \& non-strict}.
Thus, \Cref{lem:strict-noR-nonstrict} also establishes \noTransformAddOn{\D \& strict}{(\setting{\& simple})}{R}{\D \& proper}{(\setting{\& simple})} and \noTransformAddOn{\D \& strict}{(\setting{\& simple})}{R}{\D \& non-strict}{ (\setting{\& simple})}.

Next, we separate \setting{\D \& proper} (and thus \setting{\D \& non-strict}) from \setting{\D \&  non-strict \& simple}, demonstrating that more than one label per edge is required to capture the full range of directed non-strict reachability.
As before, we utilize the structure of a directed triangle, constructing a cycle with two overlapping triangles in the reachability graph. 
\begin{lemma}[\noTransform{\D \& proper}{R}{\D \& non-strict \& simple}]
\label{lem:proper-noS-nonstrict_simple}
    There is a graph in the \setting{\D \& proper} setting whose reachability graph cannot be obtained from a graph in the \setting{\D \& non-strict \& simple} setting.
\end{lemma}
\ifshort
\else
\begin{proof}
    Consider the following temporal graph \gcal in the \setting{\D \& proper} setting (left) and the corresponding reachability graph (right). 
    For the sake of contradiction, let \hcal be a temporal graph in the \setting{\D \& non-strict \& simple} setting whose reachability graph is isomorphic to that of \gcal.
    \begin{figure}[h]
        \centering
        \includegraphics[width=0.5\linewidth]{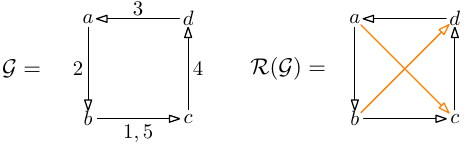}
    \end{figure}
    Observe that by \Cref{lem:nonstrict-triangles}, we cannot take one of the orange edges $(a,c)$ and $(b,d)$ into \hcal without taking out at least one black edge of the corresponding triangle $(a,c,d)$ and $(b,d,a)$. 

    Now, observe that $d$ reaches only vertex $a$, so \hcal needs to include the direct edge $(d,a)$. Similarly, $b$ is reached only by $a$, so \hcal needs to include $(a,b)$. This implies that we cannot use the orange edge $(b,d)$.
    Next, for $c$ to reach $d$, \hcal needs to include the edge $(c,d)$, which implies that we also cannot use the orange edge $(a,c)$.
    Thus, for $a$ to reach $c$, \hcal has to include $(b,c)$ and the footprint of \hcal has to be isomorphic to the footprint of \gcal.
    
    The reachabilities imply that $\lambda(a,b)<\lambda(d,a)<\lambda(c,d)$ and also $\lambda(c,d)<\lambda(b,c)<\lambda(a,b)$. But this requires $\lambda(a,b)<\lambda(a,b)$; a contradiction.
\end{proof}
\fi
Note that this also implies \noTransform{\D \& non-strict}{R}{\D \& non-strict \& simple}, \noTransform{\D \& proper}{R}{\D \& proper \& simple}, and  \noTransform{\D \& non-strict}{R}{\D \& proper \& simple}.

With the final reachability separation, we show that \setting{\D \& proper \& simple} is a true subset of \setting{\D \& non-strict \& simple}. This separation involves a more intricate construction: three vertices are fully connected (clique) at two distinct time steps, while vertices $a$ and $c$ must traverse through these cliques to reach $b$ and $d$. This traversal, however, has to happen at the different time steps so that $a$ reaches both $b$ and $d$, and $c$ reaches only $d$. 
Additionally, carefully placed directed paths are used to forbid certain edges (shortcuts) in the footprint. 
\begin{lemma}[\noTransform{\D \& non-strict \& simple}{R}{\D \& proper \& simple}]
\label{lem:proper_simple-noS-nonstrict_simple}
    There is a graph in the \setting{\D \& non-strict \& simple} setting whose reachability graph cannot be obtained from a graph in the \setting{\D \& proper \& simple} setting.
\end{lemma}
\begin{proof}
\setcounter{ppclaim}{0} 
    Consider the following temporal graph $\gcal$ in the $\setting{\D \& non-strict \& simple}$ setting. 
    \begin{figure}[ht]
        \centering
        \includegraphics[width=0.85\linewidth]{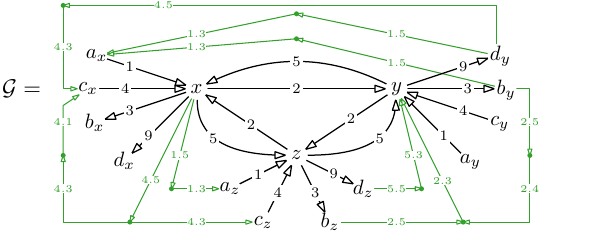}
    \end{figure}
    For sake of readability, the following edges have been omitted in the illustration, with each set having one representative shown in the illustration in green.
    Let $C=\{x,y,z\}$.
    \begin{itemize}
        \item $E_1=\{(i,h_{a_j}^i,1.5), (h_{a_j}^i,a_j,1.3), (i,h_{c_j}^i,4.5), (h_{c_j}^i,c_j,4.3) \colon i\in C \text{ and } j\in C\setminus\{i\}\}$;
        \item $E_2=\{(b_j,h_{b_j}^i,2.5), (h_{b_j}^i,i,2.3), (d_j,h_{d_j}^i,5.5), (h_{d_j}^i,i,5.3) \colon i\in C \text{ and } j\in C\setminus\{i\}\}$; 
        \item $E_3=\{(b_j,h_{b_j}^{a_i},1.5), (h_{b_j}^{a_i},a_i,1.3),(d_j,h_{d_j}^{a_i},1.5), (h_{d_j}^{a_i},a_i,1.3), (d_j,h_{d_j}^{c_i},4.5), (h_{d_j}^{c_i},c_i,4.3) \colon$\\
        \mbox{$i\in C$} and $j\in C\setminus\{i\}\}$.
        \item $E_4=\{(b_i,s,2.7), (s,h_{b_j}^i,2.5), (h_i^{c_j},s,4.3), (s,c_j,4.1) \colon i\in C \text{ and } j\in C\setminus\{i\}\}$;
    \end{itemize}
    We refer to the vertices $h^\alpha_\beta$ as \textit{helper} vertices and to $\{x,y,z\}$ as the \textit{center} vertices. The edges in $E_1$, $E_2$, $E_3$, and $E_4$ form paths of length two in the footprint and the time labels are chosen such that they do not form temporal paths.
    Specifically, for $i\in\{x,y,z\}$ and $j\in\{x,y,z\}\setminus\{i\}$, 
    $E_1$ forms paths from $i$ to $a_j$ and $c_j$; 
    $E_2$ forms paths from $b_i$ and $d_i$ to $j$; 
    $E_3$ forms paths from $b_i$ and $d_i$ to $a_j$ and $c_j$; and
    $E_4$ forms paths from $b_i$ to the helper vertices of $b_j$, and from the helper vertices of $c_j$ to $c_i$.

    For the sake of contradiction, let $\mathcal{H}$ be a temporal graph in the $\setting{\D \& proper \& simple}$ setting whose reachability graph is isomorphic to that of $\mathcal{G}$.
    Observe that the center vertices form a strongly connected component in $G_2$ and in $G_5$, and recall from \Cref{lem:proper-cliques} that in the $\setting{\D \& proper}$ setting, we need at least four temporal edges to fully connect three vertices.
    Further, observe the following reachabilities in \gcal: \begin{itemize}
        \item $x,y,z$ reach every $b_i$ and $d_i$ with $i\in\{x,y,z\}$;
        \item for $i\in\{x,y,z\}$, every $a_i$ reaches $x,y,z$, $b_j$ and $d_j$ with $j\in\{x,y,z\}$ and
        \item every $c_i$ reaches $x,y,z$ and $d_j$ for $j\in\{x,y,z\}$;
        \item every helper vertex $h^\alpha_\beta$ can reach the vertex $\beta$, 
        but not $\alpha$.
        Conversely, $\alpha$ can reach $h^\alpha_\beta$, but not $\beta$. However, $\beta$ can reach $\alpha$.
    \end{itemize}
    First, we show that the construction of the helper vertices $h^\alpha_\beta$ forbids the direct edges $(\beta,\alpha)$ in \hcal for all but $(c_j,i)$ and $(i,b_j)$.
    \begin{ppclaim} \label{claim:alien-no-direct-edge}
        In \hcal, there can be no edge from $F=\{(a_j,i),(i,d_j),(a_i,b_j),(a_i,d_j),(c_i,d_j),\allowbreak(h_{b_j}^i,b_i),(c_j,h_i^{c_j}) \colon i,j\in\{x,y,z\}\text{, }j\neq i\}$.
    \end{ppclaim}
    \begin{claimproof}[Proof of \Cref{claim:alien-no-direct-edge}]
        Recall that any labeling of a directed triangle in a $\setting{non-strict}$ or $\setting{proper}$ setting leads to at least one transitive reachability (\Cref{lem:nonstrict-triangles}), and 
        observe that every direct edge $(\beta,\alpha)$ would form a directed triangle with the corresponding green edges $(\alpha,h_{\beta}^\alpha),(h_{\beta}^\alpha,\beta)$. 
        Let $(\beta,\alpha)\in F$.
        Observe that $h_\beta^\alpha$ reaches $\beta$, but reaches no other vertex that also reaches $\beta$.
        Thus, the direct edge $(h_{\beta}^\alpha, \beta)$ must be included in the footprint $H$ of \hcal.
        Furthermore, $h_\beta^\alpha$ is reached by  $\alpha$ but by no other vertex that $\alpha$ can reach. 
        Hence, the direct edge $(\alpha,h^\alpha_{\beta})$ must also be in $E(H)$.
        As a result, $(\beta,\alpha)\notin E(H)$, since $\alpha,\beta,h^\alpha_\beta$ form a directed triangle in $\rcal(\gcal)$, which is impossible if $H$ contains the edges of the triangle.
    \end{claimproof}    
    From \Cref{claim:alien-no-direct-edge} follows that for $a_i$ to reach $d_j$ in \hcal, no shortcut edge $(a_i,j)$, $(i,d_j)$, or $(a_i,d_j)$ can be used. Additionally, $a$ reaches no vertex that also reaches $d_j$ in \gcal. 
    Thus, $a_i$~must use a temporal path with support $(a_i,i,\dots,j,d_j)$ in \hcal, and $(a_i,i),(j,d_j)\in E(H)$.
    As this holds for all $i$ and $j$, there must exist paths between all center vertices, ensuring that each $a_i$ reaches $x$, $y$ and $z$ via its respective vertex $i$. After these paths, there must also be an edge from $x$, $y$ and $z$ to their respective $d_j$.
    The following graph $\hcal'$ summarizes the current information on \hcal, with some large time step $T$. Black and green edges must be present in \hcal, while for gray edges no information is yet available.
    \begin{figure}[h]
        \centering
        \includegraphics[width=0.8\linewidth]{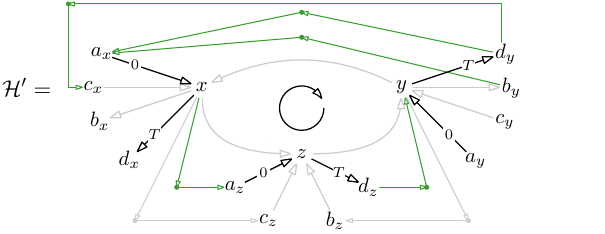}
    \end{figure}
    Next, we consider the $b$ and $c$ vertices.
    Note that the helper construction does not forbid the direct edges $(i,b_j)$ and $(c_j,i)$ for $i\neq j$ in \hcal:
    \begin{itemize}
        \item $h^{b_j}_i$ is reached only by $b_j$, requiring $(b_j,h^{b_j}_i)\in E(H)$. However, $h^{b_j}_i$ reaches $x,y,z$, $d_x,d_y,d_z$, and $b_i$ via $i$. Consequently, we could add $(h^{b_j}_i,k)$ for any $k\in\{x,y,z\}$. 
        \item $h^i_{c_j}$ only reaches $c_j$, requiring $(h^i_{c_j},c_j)\in E(H)$.
        Moreover, $h^i_{c_j}$ is reached by $x,y,z$, $a_x,a_y,a_z$, and $c_i$ via $i$. Consequently, we could add $(k,h^i_{c_j})$ for any $k\in\{x,y,z\}$. 
    \end{itemize}    
    While we cannot establish \Cref{claim:alien-no-direct-edge} for $(i,b_j)$ or $(c_j,i)$, we can prove the following slightly weaker result.
    For each $b_i$, there is exactly one center vertex (out of $x,y,z$) that has an edge to $b_i$, while the other two center vertices do not.
    Similarly, for each $c_i$, there is exactly one center vertex that $c_i$ has an edge to, while it has no edges to the remaining center vertices.
    We begin by proving this for the $b$ vertices.
    \begin{ppclaim} \label{claim:alien-bijection-b}
        There exists a bijective mapping $f\colon\{x,y,z\}\rightarrow\{x,y,z\}$ so that \mbox{$(f(i),b_i)\in E(H)$} and $(k,b_i)\notin E(H)$ for all $k\in\{x,y,z\}\setminus\{f(i)\}$.
    \end{ppclaim}
    \begin{claimproof}
        By \Cref{claim:alien-no-direct-edge}, $a_i$ cannot reach $b_j$ in \hcal via the direct edge $(a_i,b_j)$.
        Therefore, $a_i$ must traverse a temporal path starting at $i$.
        Since $b_j$ can only be reached by the center vertices $x,y,z$ and  $a_x,a_y,a_z$, it follows that $(x,b_j)$, $(y,b_j)$ or $(z,b_j)$ must be in \hcal.
        
        If $(i,b_i)\in E(H)$ and the corresponding helper edges are included in \hcal, the claim holds with $f$ defined as the identity mapping.
        Thus, it remains to prove the claim for the case where $(i,b_i)\not\in E(H)$ for at least one $i\in\{x,y,z\}$.
        Without loss of generality, assume $(z,b_z)\not\in E(H)$; the remaining cases follow by symmetry of the construction.
        
        From $(z,b_z)\notin E(H)$ follows $(y,b_z)\in E(H)$ or $(x,b_z)\in E(H)$ to ensure that $b_z$ is reachable.
        If $(y,b_z)\in E(H)$, then $(h_{b_z}^y,y)\not\in E(H)$, as otherwise, $y, b_z,h_{b_z}^y$ would form a directed triangle in both $H$ and $\rcal(\gcal)$, which is impossible in a \setting{simple \& proper} graph. Consequently, $h_{b_z}^y$ must have an edge to $x$ or to $z$.
        The same reasoning applies if $(x,b_z)\in E(H)$, implying $(h_{b_z}^x,x)\not\in E(H)$ and requiring $h_{b_z}^x$ to have an edge to $y$ or to $z$.
        
        Now, observe that the helper vertices cannot have edges to arbitrary center vertices.
        Let~the center vertex $i\in\{x,y,z\}$ have incoming edges from two helper vertices with different out-neighbors in $\mathcal{G}$, for example, $h^y_{-}$ with an edge to $y$ in \gcal, and $h^z_-$ with an edge to $z$ in~\gcal.
        These two vertices have incompatible reachabilities: $h^y_-$ must reach $b_y$ (but not $b_z$), and $h^z_-$~must reach $b_z$ (but not $b_y$).
        Thus, at most one of the two can reach their respective $b$ vertex via $i$.
        If both $h^y_-$ and $h^z_-$ had temporal paths to their $b$ vertex starting in~$i$, they could not depart $i$ at the same time, as this would result in both reaching $b_y$ and~$b_z$.
        Therefore, one temporal path must occur earlier than the other. However, the helper vertex taking the earlier path could also take the later path and would reach both $b_y$ and $b_z$.
        Consequently, one of the helper vertices must reach its respective $b$ vertex via a different center vertex. 

        This implies that, to satisfy the reachabilities of the helper vertices, all $h^i_{b_j}$ for a fixed $i$ must reach $b_i$ via the same center vertex.
        Consequently, there exists a mapping $f\colon\{x,y,z\}\rightarrow\{x,y,z\}$ such that $(h_{b_j}^i,f(i))\in E(H)$.
        
        Consider the center vertex $y$ and assume $f^{-1}(y)=x$. 
        According to this mapping, $y$ has incoming edges from $h^{x}_{b_y}$ and $h^{x}_{b_z}$.
        If $(y,b_y)\in E(H)$, then $h^{x}_{b_y},y,b_y$ would form a triangle in both $H$ and $\rcal(\hcal)$\,--\,a contradiction.
        Similarly, if $(y,b_z)\in E(H)$, then $h^{x}_{b_z},y,b_z$ would form a triangle in both $H$ and $\rcal(\hcal)$. 
        Thus, we conclude that only $(y,b_x)\in E(H)$. 
        The same holds for any other center vertex and any possible mapping $f$. 
        
        Therefore, there exists a bijective mapping $f$ on the center vertices such that $(f(i),b_i)\in E(H)$ for all $i\in\{x,y,z\}$ and $(k,b_i)\not\in E(H)$ for all $k\in\{x,y,z\}\setminus\{f(i)\}$.
    \end{claimproof}
    The proof for the $c$ vertices follows analogously to that of the $b$ vertices.
    Here, every occurrence of $b_i$ is replaced with $c_i$, and the direction of each mentioned edge is reversed. Instead of two helper edges pointing toward the same center vertex leading to a contradiction, the contradiction arises from two helper edges pointing away from the same center vertex.
    This reversal of direction results in the same type of mapping, ensuring that each $c_i$  has an edge to exactly one center vertex, with no edges to the others. The reasoning and conclusions remain identical, making a separate proof unnecessary.
    \begin{ppclaim} \label{claim:alien-bijection-c}
        There exists a bijective mapping $g\colon\{x,y,z\}\rightarrow\{x,y,z\}$ so that \mbox{$(c_i,g(i))\in E(H)$} and $(c_i,k)\notin E(H)$ for all $k\in\{x,y,z\}\setminus\{g(i)\}$.
    \end{ppclaim}
    We will now draw the final conclusion, proving that \hcal can have no \setting{ proper \& simple} labeling. 
    Let $i\in\{x,y,z\}$ and $f$ and $g$ be bijections as in \Cref{claim:alien-bijection-b} and \Cref{claim:alien-bijection-c}. 
    
    The labeling of $(a_i,i), (c_{g(i)},i), (i,b_{f(i)}),(i,d_i)$ must ensure that $a_i$ reaches both $b_{f(i)}$ and $d_i$, while $c_{f(i)}$ reaches only $d_i$.
    Consequently, $\lambda(a_i,i)<\lambda(i,b_{f(i)})<\lambda(i,c_{g(i)})<\lambda(i,d_i)$.

    $c_{g(i)}$ has to reach every $d_j, j\in\{x,y,z\}$.
    Since no shortcut edge can be used, $c_{g(i)}$ must traverse a temporal path with support $(c_{g(i)},i,\dots,j,d_j)$.
    This traversal must happen after the times $\lambda(c_{g(i)},i)$, and, consequently, after $\lambda(i,b_{f(i)})$.
    Thus, the $a$ and the $c$ vertices must travel at different times, requiring an $x,y,z$-clique to exist at two different points in time.
    Each of these requires four temporal edges between $x$, $y$, and $z$ in a \setting{proper} graph; eight in total. However, at most six temporal edges can exist between $x,y,z$ in a \setting{simple} graph, a contradiction.

\end{proof}

\subsection{Support separations}
We present the remaining separations that hold only under support equivalence, namely the separating structures that show that \setting{\D \& strict \& simple} is strictly more expressive than \setting{\D \& non-strict} and \setting{\D \& proper}. Note that these do not hold under reachability equivalence because of the transformation presented in \Cref{subsec:saturation}.

First, we separate \setting{\D \& non-strict} (\setting{\& simple}) from \setting{\D \& strict \& simple} using the (by now well known) directed triangle.
\begin{lemma}[\noTransform{\D \& non-strict \& simple}{S}{\D \& strict \& simple}]
    There is a graph in the \setting{\D \& non-strict \& simple} setting such that there is no support equivalent graph in the \setting{\D \& strict \& simple} setting.
\end{lemma}
\ifshort
\else
\begin{proof}
    Consider the following simple temporal graph \gcal (left) in the \setting{non-strict} setting and the corresponding reachability graph (right).
    For the sake of contradiction, let \hcal be a support equivalent temporal graph in the \setting{strict \& simple} setting.
    \begin{figure}[ht]
        \centering
        \includegraphics[width=0.5\linewidth]{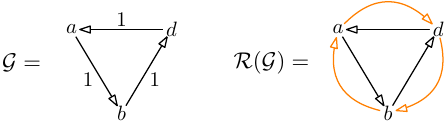}
    \end{figure}
    To preserve the support of $a$ to $c$ via $b$, \hcal needs $\lambda(a,b)<\lambda(b,c)$. Similarly, for $b$ to reach $a$ via $c$ and $c$ to reach $b$ via $a$, we get $\lambda(a,b)<\lambda(b,c)<\lambda(c,a)<\lambda(a,b)$, a contradiction.
\end{proof}
\fi

The separation of \setting{\D \& proper} from \setting{\D \& strict \& simple}, is achieved by the same structure separating \setting{\D \& proper} from \setting{\D \& non-strict \& simple} (\Cref{lem:proper-noS-nonstrict_simple}). 
\begin{lemma}[\noTransform{\D \& proper}{S}{\D \& strict \& simple}]\label{lem:noTransform-proper-S-strict_simple}
    There is a graph in the \setting{\D \& proper} setting such that there is no support equivalent graph in the \setting{\D \& strict \& simple} setting.
\end{lemma}
\ifshort
\else
\begin{proof}
    Consider the following temporal graph \gcal in the \setting{\D \& proper} setting (left) and the corresponding reachability graph (right). 
    For the sake of contradiction, let \hcal be a support equivalent temporal graph in the \setting{\D \& strict \& simple} setting.
    \begin{figure}[ht]
        \centering
        \includegraphics[width=0.5\linewidth]{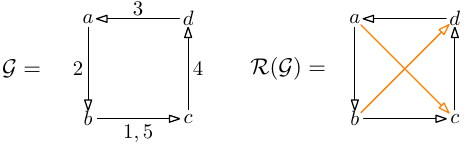}
    \end{figure}
    
    Observe that we require all black edges in \hcal to preserve the support of the direct reachabilities. Furthermore, there are two transitive reachabilities, namely $a$ reaches $c$ via $b$ and $b$ reaches $d$ via $c$. Now observe that $c$ not reaching $a$ and $d$ not reaching $b$ imply $\lambda(c,d)>\lambda(d,a)>\lambda(a,b)$.
    There are two options for the label of $(b,c)$.
    Either $\lambda(b,c)>\lambda(a,b)$ and $\lambda(b,c)>\lambda(c,d)$, which enables $a$ to reach $c$ via $b$ and avoids $a$ reaching $d$,
    or $\lambda(b,c)<\lambda(a,b)<\lambda(c,d)$ which enables $b$ to reach $d$ via $c$.
    However, in the first case, $a$ cannot reach $c$ via $b$ and in the second, $b$ cannot reach $d$ via $c$.

    Adding some of the orange edges to \hcal will not achieve the transitive reachabilities, so there can be no support equivalent labeling.
\end{proof}
\fi
Since \setting{\D \& proper} is a subset of \setting{\D \& strict}, \Cref{lem:noTransform-proper-S-strict_simple} also shows that \noTransform{\D \& strict \& simple}{S}{\D \& strict}.

\subsection{Transformations}
We present the transformations for the directed temporal settings.

\ifshort
\else
First, in \Cref{subsubsec:dilation}, we generalize the \emph{dilation} technique \cite{casteigts_simple_2024} to also work for directed graphs. We call it \textit{support-dilation}, as it yields a support-preserving transformation from \setting{\D \& non-strict} to \setting{\mbox{\D \& proper}}. As it is a direct generalization of dilation, it can also lead to a blow-up of the graphs lifetime of up to $(n-1)\cdot(\Delta+1)\cdot\tau$, where $\Delta$ is the maximum degree of the graph.
We then present a different transformation, which we call \textit{reachability-dilation}, which is only reachability-preserving, but constructs a graph with at most twice as many labels, which is significantly smaller than the blow-up of support-dilation.

Then, in \Cref{subsec:saturation}, we present the reachability-preserving transformation from \textbf{any} directed setting into the \setting{\D \& strict \& simple} setting.

Lastly, in \Cref{subsubsec:semaphore}, we discuss the generalization of the \emph{Semaphore} construction, which transforms the \setting{\D \& strict} setting into the \setting{\D \& proper \& simple} setting.
\fi

\subsubsection{Support-Dilation (\yesTransform{non-strict}{S}{proper}) and Reachability-Dilation (\yesTransform{non-strict}{R}{proper})} \label{subsubsec:dilation}
Before analyzing directed dilation, we compare the structure of the snapshots of \setting{non-strict} temporal graphs in an \setting{undirected} and \setting{directed} setting.
In undirected graphs, each snapshot consists of one to $n$ connected components, each forming a clique in
\begin{figure}[ht]
    \centering
    \includegraphics[width=0.7\linewidth]{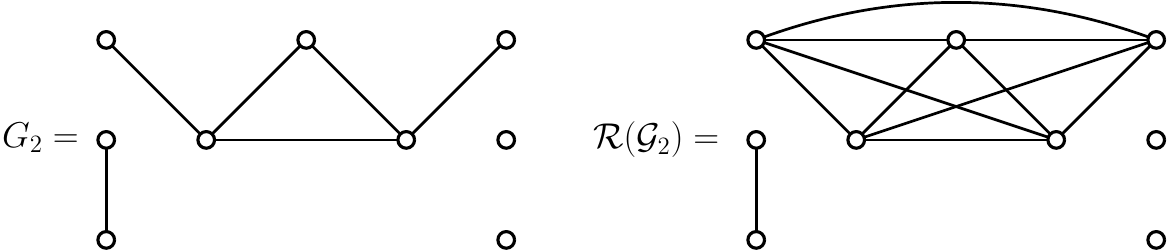}
    \caption{A cutout of some undirected temporal graph at time $2$ with four connected components (left) and the reachability graph of the subgraph at time $2$ (right).
    }
    \label{fig:undirected-snapshot-reachability}
\end{figure}
the reachability graph.
In contrast, snapshots of directed temporal graphs consist of one to $n$ weakly connected components, which can be interpreted as a directed acyclic graph (DAG). Such a DAG has as vertices the strongly connected subgraphs (possibly of size 1) of the weakly connected component connected by directed edges in an acyclic manner.
In the reachability graph, each strongly connected component forms a clique, while the DAG forms its \begin{figure}[ht]
    \centering
    \includegraphics[width=0.7\linewidth]{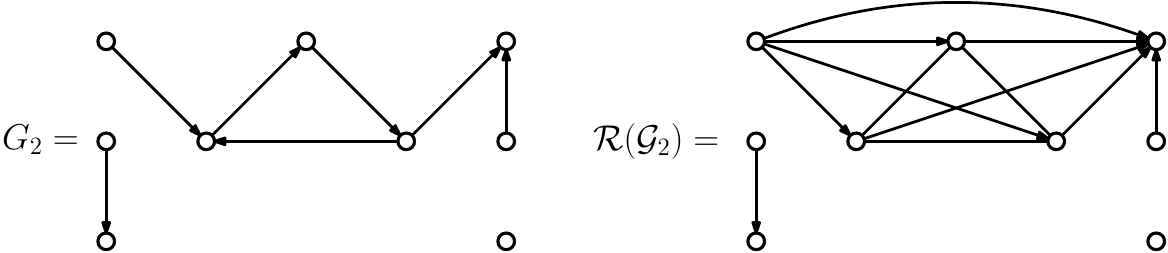}
    \caption{A cutout of some directed temporal graph at time $2$ with three weakly connected components (left) and the reachability graph of the subgraph at time $2$ (right).}
    \label{fig:directed-snapshot-reachability}
\end{figure}
transitive closure.
Now, comparing these structures, we observe that undirected snapshots represent a special case of directed snapshots. This insight allows us to generalize the undirected dilation technique introduced in \cite{casteigts_simple_2024}.

\textit{Dilation} transforms \setting{non-strict} graphs into \setting{proper} graphs while preserving support.
The process operates on each snapshot individually, duplicating every non-strict journey within a connected component of the snapshot, thereby stretching the time step to construct equivalent strict journeys. Subsequent snapshots are shifted accordingly.
Note that the new lifetime may be as large as $(n-1)\cdot(\Delta+1)\cdot\tau$, where $\Delta$ is the maximum degree. 

For directed graphs, we generalize the dilation process to also handle weakly connected components. We term this generalized approach \emph{support-dilation}, to distinguish it from the second dilation process we introduce in the later part of this section which is only reachability-preserving but achieves a much smaller blow-up of the lifetime.
\begin{definition}[Support-Dilation]
    Let \gcal be a temporal graph and consider each snapshot $G_t$ individually. Without loss of generality, we assume $t=1$ (otherwise, we shift the labels of the earlier snapshots in the construction).

    Every snapshot consists of weakly connected components $\mathbb{W}$ that each can be interpreted as a DAG.
    We consider each $W\in\mathbb{W}$ separately, so let $D=(\mathbb{S},E)$ be the DAG representing $W$. This DAG contains strongly connected components $\mathbb{S}$ (possibly of size 1) as vertices connected by directed edges. Any DAG can be topologically ordered, \ie the vertices can be arranged as a linear ordering that is consistent with the edge directions. Let $\ell$ be such an ordering.

    First, in order of $\ell$, we label the edges in $E$ \emph{between} the strongly connected components with distinct time labels starting with 1. Second, in order of $\ell$, we ``dilate'' the edges \emph{inside} each strongly connected component as follows (subsequent labels on $E$ are shifted accordingly):
    Let $S\in\mathbb{S}$ be a strongly connected component and $\alpha=\ell(S)$. First, we assign each edge the set of labels $\{\alpha+1,\alpha+2,\dots,\alpha+k\}$, where $k$ is the longest directed path in $S$. Now, we color the edges of $S$ such that no adjacent edges have the same color. By Vizing's theorem, this requires at most $\Delta+1$ colors, where $\Delta$ is the maximum degree of $S$. Let $c\colon E_S\rightarrow[0,\Delta]$ be such a coloring and let $0<\varepsilon<\frac{1}{\Delta+1}$. Now, for every edge $e\in E_S$, we add $c(e)\cdot\varepsilon$ to every label of $e$.    

    After processing every snapshot, support-dilation returns the adjusted temporal graph.
\end{definition}
Note that the dilation of the strongly connected components is a direct analogue of the dilation process of \cite{casteigts_simple_2024}, while the main addition is the ordering and temporal shifting of the strongly connected components within the DAG representing a weakly connected component. 
We now proceed to prove that support-dilation indeed preserves support.
\begin{theorem} [\yesTransform{\D \& non-strict}{S}{\D \& proper}]
\label{thm:support-dilation}
    Given a \setting{\D \& non-strict} temporal graph \gcal, \emph{support-dilation} transforms \gcal into a \setting{\D \& proper} temporal graph \hcal such that there is a non-strict temporal path in \gcal if and only if there is a strict temporal path in \hcal with the same support.
\end{theorem}
\begin{proof}
    (\hcal is proper.) Every snapshot is processed independently and shifted according to changes in previous snapshots. Within each snapshot, the weakly connected components are handled separately, which is valid as they are independent with respect to a proper labeling.
    Within each weakly connected component, the strongly connected components are processed independently, following the ordering determined by the corresponding DAG. Each strongly connected component is shifted according to changes occurring earlier in the DAG ordering.
    As argued by \cite{casteigts_simple_2024}, the addition of $c(e)\cdot\varepsilon$ to the labels of every edge $e$ according to the coloring $c$ in the strongly connected components ensures that components are labeled `proper'ly. Consequently, \hcal is proper.\\
    (\hcal preserves the temporal paths.) Given a weakly connected component in a snapshot $G_t$ of \gcal considered independently, let $D=(\mathbb{S},E)$ be the DAG representing it. Without loss of generality, assume $t=1$. Now, consider some strongly connected component $S\in\mathbb{S}$ at $\alpha=\ell(S)$ in the ordering of $D$. 
    The longest path in $S$ has length $k$ and in \hcal, every edge of $S$ is assigned all the labels from $\alpha+1$ to $\alpha+k$. Thus, by the same argument as in \cite{casteigts_simple_2024}, for every (non-strict) path of length $k'$ in $S$, there is a \textit{strict} temporal path in the adjustment of $S$ in \hcal along the same sequence of edges, going over the same edges but with labels $1,2,\dots,k'$ (up to addition of $c(e)\cdot\varepsilon$).
    Now, since the order of the strongly connected components in a DAG is maintained, temporal paths over multiple strongly connected components in one snapshot are preserved.
    The same holds for temporal paths spreading over multiple snapshots. 
\end{proof}
For a reachability-preserving transformation, we propose a more efficient dilation process prioritizing sparsity. Unlike support-dilation, which reconstructs every path within each strongly connected component, \emph{reachability-dilation} replaces each strongly connected component with a bidirected spanning tree
with a distinct, single label on each directed edge.
The components are then connected according to the DAG ordering as in the support-dilation process.
This approach replaces an edge between $u$ and $v$ at time step $t$ with at most two temporal edges, substantially reducing the blow-up in the graph's lifetime. The resulting lifetime is bounded by $2\cdot\tau$.
\ifshort \else 
Refer to \Cref{fig:definition-reachability-dilation} for an illustration of the construction.
\fi
\begin{definition} [Reachability-Dilation] \label{def:reachability-dilation}
    Let \gcal be a temporal graph and consider each snapshot $G_t$ individually. Without loss of generality, $t=1$.
    Let $D=(\mathbb{S},E)$ be the DAG representing a weakly connected component of $G_1$, and $\ell$ an ordering of $D$.

    First, in the order of $\ell$, we label the edges in $E$ \emph{between} the strongly connected components with distinct time labels starting with 1. Second, in the order of $\ell$, we replace the edges \emph{inside} each strongly connected component as follows (subsequent labels on $E$ are shifted accordingly):
    Let $S\in \mathbb{S}$. As $S$ is strongly connected, the underlying undirected graph contains a spanning tree $T_S$. It was shown by \cite{christiann_inefficiently_2024} that one can temporally connect any undirected tree using at most two labels per edge. Slightly adjusting their construction and argument, we temporally connect $T_S$ by turning it into a bidirected tree and placing one distinct time label one each directed edge: 
    Choose some vertex of $S$ as the arbitrary root. Starting at the leafs, assign distinct, increasing labels to the upwards edge until reaching the root. Now, starting at the root, assign further increasing labels to the downwards edges until reaching every leaf. 

    After processing every snapshot, reachability-dilation returns the adjusted temporal graph.
\end{definition}
\ifshort\else
\begin{figure}[ht]
    \centering
    \includegraphics[width=\linewidth]{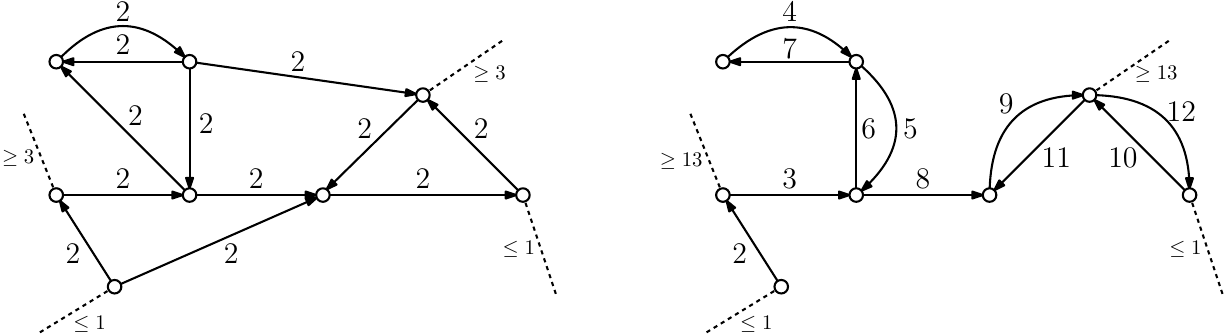}
    \caption{Illustration of the \textit{reachability-dilation} described in \Cref{def:reachability-dilation} for some \gcal with snapshot $G_2$ (left) and the proper labeling with the bidirected subtrees in the DAG structure (right). The dotted edges represent adjacent earlier and later (shifted in the proper labeling) edges.}
    \label{fig:definition-reachability-dilation}
\end{figure}
\fi
In the following, we prove that reachability-dilation is reachability-preserving. But first, we show that the modified spanning tree construction yields a temporally connected graph.
\begin{lemma} \label{lem:spanning-tree-temporally-connected}
    Any bidirected tree with root $r$, whose labels are assigned as in \Cref{def:reachability-dilation} (starting at the leafs, assign distinct, increasing labels to the upwards edge until reaching the root -- then starting at the root, assign further increasing labels to the downwards edges until reaching every leaf) is temporally connected. 
\end{lemma}
\begin{proof}
    By construction, the upwards directed edges form temporal paths from the leafs to the root $r$, and the downwards directed edges form temporal paths from $r$ to the leafs. Since all upwards paths are assigned time labels strictly smaller than all labels of the downwards paths, $r$ is a pivot vertex, \ie there is a time step $t$ such that every vertex of the tree reaches $r$ before $t$ and can be reached by $r$ after $t$. Thus, the tree is temporally connected. 
\end{proof}
\begin{theorem} \label{thm:reachability-dilation}
    Given a \setting{\D \& non-strict} temporal graph \gcal, \emph{reachability-dilation} transforms \gcal into a \setting{\D \& proper} temporal graph \hcal with $\rcal(\gcal)=\rcal(\hcal)$.
\end{theorem}
\begin{proof}
    (\hcal is proper.) This follows by the same arguments as in \Cref{thm:support-dilation} and the fact that the spanning tree labeling is proper.\\
    ($\rcal(\gcal)=\rcal(\hcal)$.) 
    Given a weakly connected component in a snapshot $G_t$ of \gcal considered independently, let $D=(\mathbb{S},E)$ be the DAG representing it. Let $S\in\mathbb{S}$ be a strongly connected component in $D$. 
    By \Cref{lem:spanning-tree-temporally-connected}, the edges in $S$ are replaced by a proper, temporally connected bidirected tree spanning all vertices of $S$. Thus, all vertices in $S$ can reach one another in \hcal within the extended time interval of $S$, so the reachability inside $S$ in \hcal is the same as in \gcal. 
    Now, since the chronological order of the strongly connected components in any DAG is maintained and the footprint of the DAG is the same, temporal reachability over multiple strongly connected components in one snapshot is the same in \hcal and \gcal. Lastly, since the chronological order of the snapshots is also maintained, the temporal reachability via multiple snapshots is also equivalent.
\end{proof}
Reachability-dilation can also be executed on undirected graphs, where it replaces each connected component of every snapshot with an undirected spanning tree. This process will be discussed in more detail in \Cref{sec:directed-VS-undirected}, \Cref{lem:yesTransform-UD_nonstrict_simple-D_proper_simple}.

\subsubsection{Saturation: \yesTransform{\D \& *} {R}{\D \& strict \& simple}} \label{subsec:saturation}
Undirected saturation is a reachability-preserving transformation from \setting{\UD \& non-strict} to \setting{\UD \& strict} temporal graphs.
While a transformation under reachability equivalence is already guaranteed by support-dilation (as \setting{\UD \& proper} is a subset of \setting{\UD \& strict}), saturation offers a simpler process with the key advantage of maintaining the same lifetime. This simplicity comes at the cost of increasing the number of edges, potentially up to $\frac{n(n-1)\tau}{2}$. 
%
The process considers each snapshot $G_1,\dots,G_\tau$ individually. For each $G_i$, it constructs the path-based transitive closure, replacing every path in $G_i$ with an edge with time label $i$. Since each connected component in a snapshot forms a clique in $\rcal(\gcal)$, that undirected edge with label $i$ does not form new reachabilities.
This simplifies the graph's structure but results a multi-labeled graph, even for small instances.

For directed graphs, an even simpler reachability-preserving transformation achieves a mapping from \textbf{any} directed setting to \setting{\D \& strict \& simple}:
Take the reachability graph as the footprint and assign a uniform time label (e.\,g., 1) to all edges.
This replaces every possible temporal path in the graph with a direct edge with label $1$, which, because of the \setting{strict} setting, yields the desired reachabilities.
It is important to note that this directed transformation does not work for undirected graphs, since reachability is not symmetric and directed edges are crucial for this construction. 
\begin{observation} [\yesTransform{\D \& *}{R}{\D \& strict \& simple}] \label{thm:saturation}
    Given a directed temporal graph \gcal in an arbitrary setting,
    $\hcal$ defined as $(\rcal(\gcal),\lambda)$ with $\lambda(e)=1$ for all $e\in E(\rcal(\gcal))$ 
    is a \setting{\D \& strict \& simple} temporal graph 
    with $\rcal(\gcal)=\rcal(\hcal)$.
\end{observation}

\subsubsection{Directed Semaphore: \yesTransform{\D \& strict}{iR}{\D \& simple \& proper}} \label{subsubsec:semaphore}
    The (undirected) semaphore technique transforms \setting{\UD \& strict} temporal graphs into \setting{\UD \& simple \& proper} graphs while preserving induced-reachability. This means it maintains the reachability between original vertices, but while introducing new vertices.
    Each undirected edge $(uv,t)$ is replaced by a ``semaphore'' gadget, which simulates the two directions of the edge separately.
    The gadget consists of two auxiliary vertices, $\overrightarrow{uv}$ and $\overrightarrow{vu}$, and four edges: $(u,\overrightarrow{uv}, t-\varepsilon)$, $(\overrightarrow{uv},v, t+\varepsilon)$, and $(v,\overrightarrow{vu}, t-\varepsilon)$, $(\overrightarrow{vu},u, t+\varepsilon)$.
    This construction guarantees a simple labeling where each gadget can be traversed in only one direction. To ensure properness of the graph, CCS used an edge coloring as in the undirected dilation. 

    The generalization to directed graphs is straightforward. Since each original edge has a direction, the construction only needs to subdivide the edge along its direction, without requiring duplication.
    Unlike the undirected construction, where each pair of vertices ends up with either zero or an even number of opposing subdivided edges (two per original edge), the directed construction may result in an odd number of such edges (one subdivided edge per original edge).
    The generalized process directly follows \cite[Theorem 2]{casteigts_simple_2024}, and we include it here for completeness. For simplicity, the construction uses fractional time labels, which can later be renormalized to integers.
\begin{definition}[Semaphore (directed)] \label{def:semaphore}
    Let \gcal be a directed temporal graph with temporal edges \ecal, let $\Delta$ be the maximum degree of the footprint $G$, and let $0 < \varepsilon < \frac{1}{2(\Delta+1)}$. 
    Consider an edge-coloring $c$ of $G$ using $\Delta+1$ colors (which exists by Vizing's theorem).
    
    Replace every directed temporal edge $e=(u,v,t)\in\ecal$ with an auxiliary vertex $\overrightarrow{uv}$ and two temporal edges $(u,\overrightarrow{uv}, t-c(e)\cdot\varepsilon)$ and $(\overrightarrow{uv},v, t+c(e)\cdot\varepsilon)$.
\end{definition}
\begin{theorem}[\yesTransform{\D \& strict}{iR}{\D \& simple \& proper}] \label{thm:semaphore}
    Given a \setting{\D \& strict} temporal graph \gcal, semaphore transforms \gcal into a \setting{\D \& simple \& proper} temporal graph \hcal such that there is $\sigma\colon V(\gcal)\rightarrow V(\hcal)$ with $(u,v)\in\rcal(\gcal)$ if and only if $(\sigma(u),\sigma(v))\in\rcal(\hcal)$.
\end{theorem}
The proof of \Cref{thm:semaphore} is almost identical to that of \cite[Theorem 2]{casteigts_simple_2024}, so we give only a brief intuition: \hcal is obviously simple, and proper by the same reasoning as for \Cref{thm:support-dilation}.
The reachability equivalence for $V(\gcal)$ follows from the construction: Each directed edge $(u,v,t)$ is subdivided into two parts. The first half is assigned a label less than $t$, and the second half is assigned a label greater than $t$ (shifted according to the edge coloring).
As a result, after traversing from $u$ to $v$ via the subdivided edge, any temporal path reaches $v$ at a time later than any first-half label of any original edge $(v,x,t)$. This ensures that a temporal path can take at most one of the original edges at each time step, and thus is strict.

As for undirected graphs, one can apply dilation and semaphore, to transform any directed setting into \setting{\D \& simple \& proper} with preserved induced-reachability. Since \setting{\D \& simple \& proper} is subset of every setting, we conclude the following.
\begin{corollary}[\setting{\D \& *} $\leftrightsquigarrow^\text{iR}$ \setting{\D \& simple \& proper}] \label{thm:iR-happy-all}
    All directed settings are induced-reachability equivalent.
\end{corollary}

\section{Comparison of Directed and Undirected Reachability} \label{sec:directed-VS-undirected}
In this section, we aim to merge the undirected and directed reachability hierarchy. Refer back to \Cref{fig:merged-hierarchy1} in the introduction for an illustration of the hierarchy. 

First, observe that directed graphs are inherently more expressive than undirected graphs, as any undirected graph can be transformed into a directed graph, but not vice versa: One simple directed edge cannot be expressed in an undirected setting under any equivalence stronger than induced-reachability. On the other hand, we can replace an undirected edge by two directed edges obtaining a bijective-equivalent directed graph, see \Cref{fig:directed-no-undirected--undirected-yes-directed} for an illustration.
\begin{figure}[ht]
    \centering
    \includegraphics[width=0.2\linewidth]{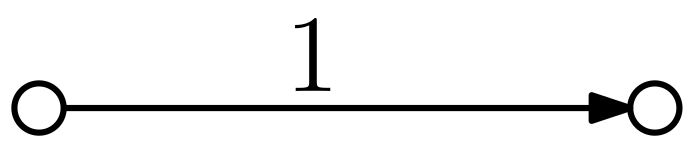} \hspace{5em}
    \includegraphics[width=0.35\linewidth]{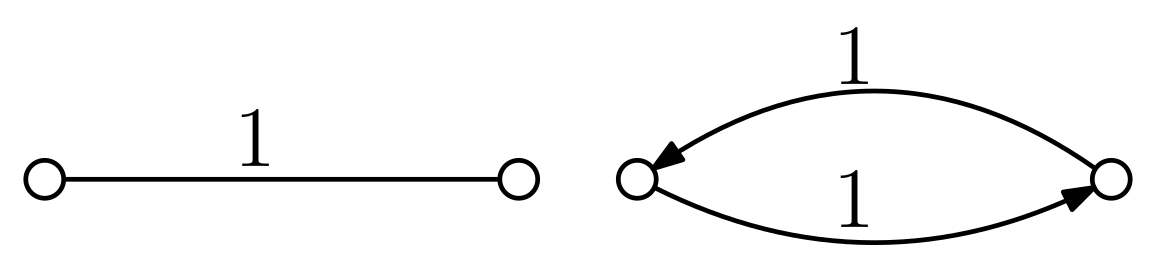}
    \caption{A separating structure from any directed setting to any undirected setting (left), and an intuition for a transformation from undirected to directed (right).}
    \label{fig:directed-no-undirected--undirected-yes-directed}
\end{figure}
This may imply that any undirected setting \setting{\UD \& x} can be reachability-preserving transformed into its directed equivalent \setting{\D \& x} using this straightforward transformation.

For both strict and non-strict settings, this approach works seamlessly, and the resulting directed graph is simple if and only if the original undirected graph was simple.
However, when transforming any undirected proper temporal graph in this way, the resulting directed graph is not proper under the standard definition. Specifically, the opposing edges $(u,v)$, $(v,u)$ derived from the same undirected edge $(uv,t)$ are assigned the same label $t$.
However, since these two edges cannot appear in the same path, their labels can be slightly adjusted (tilted) to make the graph proper without altering the reachability.
With this, we make the following first observation:
\begin{observation}
    It holds that \yesTransform{undirected \& x}{R}{directed \& x}, but the converse is not true \noTransform{directed \& x}{R}{undirected \& x}.
\end{observation}
Note that in the direct translation from undirected proper to directed, there are only temporal paths with strictly increasing labels possible. This is the essential property of proper temporal graphs. 
Therefore, we want to propose a (bijective) equivalent and more convenient definition of \emph{proper} for directed temporal graphs. 
\begin{definition}
    A directed temporal graph \gcal is proper, if for all temporal edges $(u,v,t)$, the adjacent edges $(v,x,t')$ with $x\neq v$ have different time labels $t\neq t'$.
\end{definition}

Now, the question is whether we can transform an undirected setting into a directed setting ``a level lower''. For \setting{\UD \& strict} and \setting{\UD \& strict \& simple} we show that this is not possible. Namely, we present a structure separating \setting{\UD \& strict \& simple} from \setting{\D \& non-strict}, which effectively separates the strict settings from the non-strict settings.
\begin{lemma}[\noTransform{\UD \& strict \& simple}{R}{\D \& non-strict}]
    There is a graph in the \setting{\UD \& strict \& simple} setting such that there is no reachability equivalent graph in the \setting{\D \& non-strict} setting.
\end{lemma}
\ifshort
\else
\begin{proof}
    Consider the following temporal graph \gcal in the \setting{\UD \& strict \& simple} setting (left) and the corresponding reachability graph (right). 
    \begin{figure}[ht]
        \centering
        \includegraphics[width=0.5\linewidth]{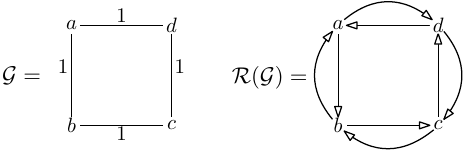}
    \end{figure}
    For the sake of contradiction, let \hcal be a temporal graph in the \setting{\D \& non-strict} setting whose reachability graph is isomorphic to that of \gcal.
    Note that in \gcal, $a$ reaches both $b$ and $d$, while $d$ does not reach $b$. Therefore, the direct edge $(a,b)$ must be included in \hcal to ensure that $a$ can reach $b$. By the same reasoning, every edge in $\rcal(\gcal)$ has to be included in \hcal.

    Consider the cycle $a,b,c,d,a$. Regardless of how the edges on this cycle are labeled, \Cref{lem:nonstrict-triangles} implies that this will create at least one transitive reachability. For example, the labeling $(a,b,4)$, $(b,c,3)$, $(c,d,2)$, $(d,a,1)$, results in $d$ reaching $b$.
    However, no such transitive reachability exists in $\rcal(\gcal)$.
    Therefore, there exists no valid labeling for \hcal.
\end{proof}
\fi
However, among the non-strict or proper settings, there exists indeed one such ``level-breaking'' transformation. Applying the reachability-dilation (\Cref{def:reachability-dilation}) to \setting{\UD \& non-strict \& simple} graphs, produces a reachability equivalent \setting{\D \& proper \& simple} graph.
Reachability dilation on an undirected graph replaces every connected component within each snapshot by a spanning tree with bidirected edges. The proof follows from the proof for directed reachability dilation (\Cref{thm:reachability-dilation}).
Refer to \Cref{fig:reachability-dilation-undirected_simple} for an illustration.
\begin{lemma}[\yesTransform{\UD \& non-strict \& simple}{R}{\D \& proper \& simple}]\label{lem:yesTransform-UD_nonstrict_simple-D_proper_simple}
    Given an \setting{\UD \& non-strict \& simple} temporal graph \gcal, \emph{reachability dilation} transforms \gcal into a \setting{\D \& proper \& simple} temporal graph \hcal with $\rcal(\gcal)=\rcal(\hcal)$.
\end{lemma}
\ifshort
\else
\begin{proof}
    (\hcal is proper.) This follows from the same arguments as in \Cref{thm:reachability-dilation}.\\
    (\hcal is simple.) Reachability-dilation replaces each connected component of a snapshot of \gcal with a spanning tree of the underlying footprint. Since the component is connected, such a tree always exists.
    By this process, each edge $uv$ within the component is either removed or replaced by two opposing directed temporal edges, each assigned one time label. As \gcal is simple, no pair of vertices will be considered more than once, and consequently, in \hcal, each pair of vertices is connected by either zero or two opposing directed, single labeled edges.\\
    ($\rcal(\gcal)=\rcal(\hcal)$.) This follows from the same arguments as in \Cref{thm:reachability-dilation}.
\end{proof}
\fi
\begin{figure}[h]
    \centering
    \includegraphics[width=0.9\linewidth]{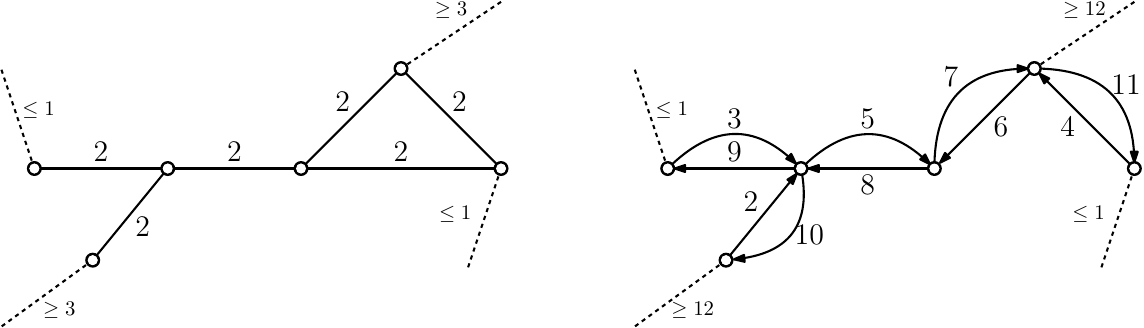}
    \caption{Illustration of the \textit{reachability-dilation} (\Cref{def:reachability-dilation}) applied to an \setting{\UD \& non-strict \& simple} graph \gcal with snapshot $G_2$ (left) and the proper labeling with the bidirected subtrees (right).
    The dotted edges represent adjacent earlier and later (shifted in the proper labeling) edges.}
    \label{fig:reachability-dilation-undirected_simple}
\end{figure}

There only remains one open question: whether \setting{\UD \& non-strict} can be transformed into \setting{\D \& proper \& simple}, into \setting{\D \& non-strict \& simple} or none of the two.
If that were the case, it would mean that when expressing undirected temporal graphs with directed temporal graphs, one requires not more than one label per edge. So to speak, one time label per direction of each edge.
One open question remains: whether \setting{\UD \& non-strict} graphs can be transformed into \setting{\D \& proper \& simple}, into \setting{\D \& non-strict \& simple}, or neither.
If such a transformation actually exists, it would imply that expressing undirected temporal graphs using directed temporal graphs requires no more than one time label per edge\,—\,essentially, one label per direction of each edge.

\section{Conclusion and Future Work} \label{sec:conclusion}
In this work, we extended the study by \cite{casteigts_simple_2024} from undirected to directed temporal graphs, analyzing how different temporal definitions impact the reachability structures that can be expressed.
We compared the temporal graph classes arising from these preliminary definitions under the equivalence notions introduced by \cite{casteigts_simple_2024}, which preserve all paths, preserve reachability, or preserve reachability when allowing auxiliary vertices.
Overall, we completed the undirected hierarchy (largely resolved by \cite{casteigts_simple_2024}), fully resolved the directed hierarchy, and nearly completed the merged hierarchy that compares the two.
Our results highlight both similarities and key differences between directed and undirected temporal graphs.

We highlight our results under reachability equivalence. We showed that \setting{directed \& strict \& simple} graphs, which are as expressive as \setting{directed \& strict} graphs, form the most expressive class: every reachability graph can be expressed using a \setting{directed \& strict \& simple} graph. At the other extreme, \setting{undirected \& proper \& simple} is the least expressive class.
Notably, any class of directed graphs is strictly more expressive than its undirected counterpart, while any class of strict graphs is strictly more expressive than any non-strict class. 

Our analysis leaves two questions open.
\begin{openquestion}[\yesTransform{UD \& non-strict}{R}{D \& proper \& simple}?] \label{oq:proper}
    Does there exist an \setting{undirected \& non-strict} graph such that there is no reachability equivalent \setting{directed \& proper \& simple} graph, or is there a reachability-preserving transformation?
\end{openquestion}
If this question is answered negatively, the following question remains to be answered:
\begin{openquestion}[\yesTransform{UD \& non-strict}{R}{D \& non-strict \& simple}?]\label{oq:nonstrict}
    Does there exist an \setting{undirected \& non-strict} graph such that there is no reachability equivalent \setting{directed \& non-strict \& simple} graph, or is there a reachability-preserving transformation?
\end{openquestion}
We have shown that every \setting{undirected \& \underline{strict}} (\setting{\& simple}) graph can be transformed into a \setting{directed \& strict \& simple} graph. This demonstrates that representing a \setting{strict} undirected graph with a directed graph while preserving reachability never requires more than one label per edge.
Now, if there exists a transformation for either \Cref{oq:proper} or \Cref{oq:nonstrict}, the same conclusion would extend to \setting{\underline{non-strict}} graphs.

\paragraph*{Future work}
We propose five directions for future research, building on the findings of this work and the contributions by \cite{casteigts_simple_2024}.
First, we want to mention questions 5 and 6 raised in \cite{casteigts_simple_2024} which dive deeper into the reachabilities each temporal graph setting can express. The first addresses computational complexity:
\begin{question}
    Given a static directed graph and a temporal graph class (setting), how computationally hard is it to decide whether the static graph can be realized as the reachability graph of a temporal graph in the given setting?
\end{question}
The second asks for a structural characterization of the reachability graphs that are realizable in each setting:
\begin{question}
    Given a temporal graph class, characterize the class of static directed graphs that can appear as the reachability graph of a temporal graph in the given class.
\end{question} 
This question is partially addressed by the separating structures presented in this work and in \cite{casteigts_simple_2024}, which identify subgraphs (e.\,g., the directed triangle) that are realizable in one setting but not another.

The next two question focus on extending the applicability of the current results. 
In this work, we primarily aimed to show the existence of transformations between settings. A natural next step is to investigate the efficiency of these transformations:
\begin{question}
    For the existing transformations, what constitutes an ``optimal'' transformed graph? For example, can we minimize the number of labels required? Are there identifiable bounds on the size of the resulting  transformed graph?
\end{question}
And lastly, we ask for a unified approach in future work:
\begin{question}
    What computational problems are independent under the existing transformations, respectively under the equivalence notions?
\end{question}
As an example, note that the computation of maximum flow in temporal graphs is independent under the semaphore transformation.
In a similar fashion, one can observe that many hardness results and algorithms are typically established for one specific temporal graph setting and then generalized to the other settings with minimal to no changes, e.\,g., shortest paths \cite{xuan_computing_2002}, connected components \cite{bhadra2003complexity}, exact edge-cover \cite{deligkas2024linespaintcityexact}, or Menger's theorem \cite{berman1996vulnerability}. This can be formalized by proving the independence of the computational problem under a fitting transformation or equivalence notion. 

\bibliography{other}

\begin{thebibliography}{21}
\providecommand{\natexlab}[1]{#1}
\providecommand{\url}[1]{\texttt{#1}}
\expandafter\ifx\csname urlstyle\endcsname\relax
  \providecommand{\doi}[1]{doi: #1}\else
  \providecommand{\doi}{doi: \begingroup \urlstyle{rm}\Url}\fi

\bibitem[Berman(1996)]{berman1996vulnerability}
Kenneth~A Berman.
\newblock Vulnerability of scheduled networks and a generalization of menger's theorem.
\newblock \emph{Networks: An International Journal}, 28\penalty0 (3):\penalty0 125--134, 1996.

\bibitem[Bhadra and Ferreira(2003)]{bhadra2003complexity}
Sandeep Bhadra and Afonso Ferreira.
\newblock Complexity of connected components in evolving graphs and the computation of multicast trees in dynamic networks.
\newblock In \emph{Ad-Hoc, Mobile, and Wireless Networks: Second International Conference, ADHOC-NOW2003, Montreal, Canada, October 8-10, 2003. Proceedings 2}, pages 259--270. Springer, 2003.

\bibitem[Bumby(1981)]{bumby1981problem}
Richard~T Bumby.
\newblock A problem with telephones.
\newblock \emph{SIAM Journal on Algebraic Discrete Methods}, 2\penalty0 (1):\penalty0 13--18, 1981.

\bibitem[Casteigts and Corsini(2023)]{casteigts_search_2023}
Arnaud Casteigts and Timothée Corsini.
\newblock In search of the lost tree: {Hardness} and relaxation of spanning trees in temporal graphs, December 2023.
\newblock URL \url{http://arxiv.org/abs/2312.06260}.
\newblock arXiv:2312.06260 [cs].

\bibitem[Casteigts et~al.(2012)Casteigts, Flocchini, Quattrociocchi, and Santoro]{casteigts_time-varying_2012}
Arnaud Casteigts, Paola Flocchini, Walter Quattrociocchi, and Nicola Santoro.
\newblock Time-{Varying} {Graphs} and {Dynamic} {Networks}, February 2012.
\newblock URL \url{http://arxiv.org/abs/1012.0009}.
\newblock arXiv:1012.0009 [physics].

\bibitem[Casteigts et~al.(2022)Casteigts, Raskin, Renken, and Zamaraev]{casteigts_sharp_2022}
Arnaud Casteigts, Michael Raskin, Malte Renken, and Viktor Zamaraev.
\newblock Sharp {Thresholds} in {Random} {Simple} {Temporal} {Graphs}.
\newblock In \emph{2021 {IEEE} 62nd {Annual} {Symposium} on {Foundations} of {Computer} {Science} ({FOCS})}, pages 319--326, Denver, CO, USA, February 2022. IEEE.
\newblock ISBN 978-1-66542-055-6.
\newblock \doi{10.1109/FOCS52979.2021.00040}.
\newblock URL \url{https://ieeexplore.ieee.org/document/9719741/}.

\bibitem[Casteigts et~al.(2024{\natexlab{a}})Casteigts, Corsini, and Sarkar]{casteigts_simple_2024}
Arnaud Casteigts, Timothée Corsini, and Writika Sarkar.
\newblock Simple, strict, proper, happy: {A} study of reachability in temporal graphs.
\newblock \emph{Theoretical Computer Science}, 991:\penalty0 114434, April 2024{\natexlab{a}}.
\newblock ISSN 0304-3975.
\newblock \doi{10.1016/j.tcs.2024.114434}.
\newblock URL \url{https://www.sciencedirect.com/science/article/pii/S0304397524000495}.

\bibitem[Casteigts et~al.(2024{\natexlab{b}})Casteigts, Morawietz, and Wolf]{casteigts_et_al2024distancetotransitivity}
Arnaud Casteigts, Nils Morawietz, and Petra Wolf.
\newblock {Distance to Transitivity: New Parameters for Taming Reachability in Temporal Graphs}.
\newblock In Rastislav Kr\'{a}lovi\v{c} and Anton{\'\i}n Ku\v{c}era, editors, \emph{49th International Symposium on Mathematical Foundations of Computer Science (MFCS 2024)}, volume 306 of \emph{Leibniz International Proceedings in Informatics (LIPIcs)}, pages 36:1--36:17, Dagstuhl, Germany, 2024{\natexlab{b}}. Schloss Dagstuhl -- Leibniz-Zentrum f{\"u}r Informatik.
\newblock ISBN 978-3-95977-335-5.
\newblock \doi{10.4230/LIPIcs.MFCS.2024.36}.
\newblock URL \url{https://drops.dagstuhl.de/entities/document/10.4230/LIPIcs.MFCS.2024.36}.

\bibitem[Christiann et~al.(2024)Christiann, Sanlaville, and Schoeters]{christiann_inefficiently_2024}
Esteban Christiann, Eric Sanlaville, and Jason Schoeters.
\newblock On {Inefficiently} {Connecting} {Temporal} {Networks}.
\newblock \emph{LIPIcs, Volume 292, SAND 2024}, 292:\penalty0 8:1--8:19, 2024.
\newblock ISSN 1868-8969.
\newblock \doi{10.4230/LIPICS.SAND.2024.8}.
\newblock URL \url{https://drops.dagstuhl.de/entities/document/10.4230/LIPIcs.SAND.2024.8}.
\newblock Artwork Size: 19 pages, 1069210 bytes ISBN: 9783959773157 Medium: application/pdf Publisher: Schloss Dagstuhl – Leibniz-Zentrum für Informatik.

\bibitem[Costa et~al.(2023)Costa, Lopes, Marino, and da~Silva]{Costa2023OnCL}
Isnard~Lopes Costa, Raul Lopes, Andrea Marino, and Ana Paula~Couto da~Silva.
\newblock On computing large temporal (unilateral) connected components.
\newblock In \emph{International Workshop on Combinatorial Algorithms}, 2023.
\newblock URL \url{https://api.semanticscholar.org/CorpusID:257102754}.

\bibitem[Deligkas et~al.(2024)Deligkas, Döring, Eiben, Goldsmith, and Skretas]{DELIGKAS2024105171}
Argyrios Deligkas, Michelle Döring, Eduard Eiben, Tiger-Lily Goldsmith, and George Skretas.
\newblock Being an influencer is hard: The complexity of influence maximization in temporal graphs with a fixed source.
\newblock \emph{Information and Computation}, 299:\penalty0 105171, 2024.
\newblock ISSN 0890-5401.
\newblock \doi{https://doi.org/10.1016/j.ic.2024.105171}.
\newblock URL \url{https://www.sciencedirect.com/science/article/pii/S0890540124000361}.

\bibitem[Fluschnik et~al.(2020)Fluschnik, Molter, Niedermeier, Renken, and Zschoche]{fluschnik_temporal_2020}
Till Fluschnik, Hendrik Molter, Rolf Niedermeier, Malte Renken, and Philipp Zschoche.
\newblock Temporal graph classes: {A} view through temporal separators.
\newblock \emph{Theoretical Computer Science}, 806:\penalty0 197--218, 2020.
\newblock URL \url{https://www.sciencedirect.com/science/article/pii/S0304397519301975}.
\newblock Publisher: Elsevier.

\bibitem[Füchsle et~al.(2022)Füchsle, Molter, Niedermeier, and Renken]{fuchsle_delay-robust_2022}
Eugen Füchsle, Hendrik Molter, Rolf Niedermeier, and Malte Renken.
\newblock Delay-{Robust} {Routes} in {Temporal} {Graphs}, January 2022.
\newblock URL \url{http://arxiv.org/abs/2201.05390}.
\newblock arXiv:2201.05390 [cs].

\bibitem[Garc{\'\i}a-L{\'o}pez and Mariju{\'a}n(2012)]{garcia2012minimal}
Jes{\'u}s Garc{\'\i}a-L{\'o}pez and Carlos Mariju{\'a}n.
\newblock Minimal strong digraphs.
\newblock \emph{Discrete Mathematics}, 312\penalty0 (4):\penalty0 737--744, 2012.

\bibitem[Göbel et~al.(1991)Göbel, Cerdeira, and Veldman]{gobel_label-connected_1991}
F.~Göbel, J.~Orestes Cerdeira, and H.~J. Veldman.
\newblock Label-connected graphs and the gossip problem.
\newblock \emph{Discrete Mathematics}, 87\penalty0 (1):\penalty0 29--40, January 1991.
\newblock ISSN 0012-365X.
\newblock \doi{10.1016/0012-365X(91)90068-D}.
\newblock URL \url{https://www.sciencedirect.com/science/article/pii/0012365X9190068D}.

\bibitem[Kempe et~al.(2002)Kempe, Kleinberg, and Kumar]{kempe_connectivity_2002}
David Kempe, Jon Kleinberg, and Amit Kumar.
\newblock Connectivity and {Inference} {Problems} for {Temporal} {Networks}.
\newblock \emph{Journal of Computer and System Sciences}, 64\penalty0 (4):\penalty0 820--842, June 2002.
\newblock ISSN 0022-0000.
\newblock \doi{10.1006/jcss.2002.1829}.
\newblock URL \url{https://www.sciencedirect.com/science/article/pii/S0022000002918295}.

\bibitem[Klobas et~al.(2023)Klobas, Mertzios, Molter, Niedermeier, and Zschoche]{klobas_interference-free_2023}
Nina Klobas, George~B. Mertzios, Hendrik Molter, Rolf Niedermeier, and Philipp Zschoche.
\newblock Interference-free walks in time: temporally disjoint paths.
\newblock \emph{Autonomous Agents and Multi-Agent Systems}, 37\penalty0 (1):\penalty0 1, June 2023.
\newblock ISSN 1387-2532, 1573-7454.
\newblock \doi{10.1007/s10458-022-09583-5}.
\newblock URL \url{https://link.springer.com/10.1007/s10458-022-09583-5}.

\bibitem[Kunz et~al.(2023)Kunz, Molter, and Zehavi]{kunz_which_2023}
Pascal Kunz, Hendrik Molter, and Meirav Zehavi.
\newblock In {Which} {Graph} {Structures} {Can} {We} {Efficiently} {Find} {Temporally} {Disjoint} {Paths} and {Walks}?, January 2023.
\newblock URL \url{http://arxiv.org/abs/2301.10503}.
\newblock arXiv:2301.10503 [cs].

\bibitem[Mertzios et~al.(2023)Mertzios, Molter, Renken, Spirakis, and Zschoche]{mertzios_complexity_2023}
George~B. Mertzios, Hendrik Molter, Malte Renken, Paul~G. Spirakis, and Philipp Zschoche.
\newblock The {Complexity} of {Transitively} {Orienting} {Temporal} {Graphs}, July 2023.
\newblock URL \url{http://arxiv.org/abs/2102.06783}.
\newblock arXiv:2102.06783 [cs].

\bibitem[Vernet et~al.(2021)Vernet, Drozdowski, Pigné, and Sanlaville]{vernet_theoretical_2021}
Mathilde Vernet, Maciej Drozdowski, Yoann Pigné, and Eric Sanlaville.
\newblock A theoretical and experimental study of a new algorithm for minimum cost flow in dynamic graphs.
\newblock \emph{Discrete Applied Mathematics}, 296:\penalty0 203--216, June 2021.
\newblock ISSN 0166-218X.
\newblock \doi{10.1016/j.dam.2019.12.012}.
\newblock URL \url{https://www.sciencedirect.com/science/article/pii/S0166218X19305554}.

\bibitem[Xuan et~al.(2002)Xuan, Ferreira, and Jarry]{xuan_computing_2002}
B.~Xuan, Afonso Ferreira, and Aubin Jarry.
\newblock Computing shortest, fastest, and foremost journeys in dynamic networks.
\newblock October 2002.

\end{thebibliography}

\appendix
\end{document}